\documentclass[twocolumn]{article}
\usepackage{geometry}
\usepackage{titlesec}
\titleformat{\section}
  {\normalfont\fontsize{13}{13}\bfseries}{\thesection}{1em}{}
\titleformat{\subsection}
  {\normalfont\fontsize{11}{11}\bfseries}{\thesubsection}{1em}{}

\renewcommand\abstract{%
}

\usepackage[T1]{fontenc}

\usepackage{url}

\usepackage[inline]{enumitem}

\usepackage{amsmath}
\usepackage{mathtools}
\usepackage{amsthm}

\usepackage{bbold}

\usepackage{xcolor}

\usepackage{tikz}
\usetikzlibrary{mindmap,calc}
\usetikzlibrary{decorations.pathmorphing}
\usetikzlibrary{positioning}
\usepackage[capitalize,nameinlink]{cleveref}

\usepackage{algorithm}
\usepackage{algorithmic}

\usepackage{thmtools}
\usepackage{thm-restate}

\declaretheorem[name=Theorem]{theorem}

\usepackage{nicematrix}
\usepackage{kbordermatrix}

\makeatletter                                                                    
\@ifpackagelater{nicematrix}{2022/01/01}                                               
{                                                                                
 \newcommand{\matrixOfConstrainsFullExpr}{%
 $$
  \constrmatrix \coloneqq 
  \begin{bNiceMatrix}[first-col, first-row, nullify-dots, code-for-first-col =
  \scriptstyle, code-for-first-row = \scriptstyle]
   & \pcand & \dcand & \ecand_1 & \Cdots & \ecand_{3\setcovsize}  & \fcand_1 & \cand_1 & \Cdots & \fcand_{3\setcovsize} & \cand_{3\setcovsize}\\
   1               &    &    & & & & h & 1 &        &   &   \\
   2               &    &    & & & & 1 & h &        &   &   \\
   \Vdots          &    &    & & & &   &   & \Ddots &   &   \\
   6\setcovsize -1 &    &    & & & &   &   &        & h & 1 \\
   6\setcovsize    &    &    & & & &   &   &        & 1 & h \\
   6\setcovsize +1 & 3t & 3t & 1 & \Cdots & 1 &   &   &        &   &   \\
   6\setcovsize +2 & 3t & 3t & 1 & \Cdots & 1 &   &   &        &   &   \\
   6\setcovsize +3 & 3t & 3t & 1 & \Cdots & 1 &   &   &        &   &   \\
   6\setcovsize +4 &    &    & \Block{3-3}{B} & & &   &   &        &   &  \\
   \Vdots          &    &    & & & &   &   &        &   &   \\
   9\setcovsize +2 &    &    & & & &   &   &        &   & 
   \CodeAfter
    \tikz \draw [dashed, shorten > = 2pt, shorten < = 2pt] (1-|6) -- (last-|6);
    \tikz \draw [dashed, shorten > = 2pt, shorten < = 2pt] (1-|3) -- (last-|3);
    \tikz \draw [dashed,shorten > = 2pt, shorten < = 2pt] (9-|1) -- (9-|last) ;
  \end{bNiceMatrix}
  $$
 }                                                     
 \newcommand{\advantageMatrixFullExpr}{%
 $$
  \advmatrix \coloneqq
  \begin{small}
  \setlength{\arraycolsep}{3pt}
  \begin{bNiceMatrix}[first-col, first-row, nullify-dots, code-for-first-col =
  \scriptstyle, code-for-first-row = \scriptstyle]
   & \pcand & \dcand & \ecand_1 & \Cdots & \ecand_{3\setcovsize} & \fcand_1 & \cand_1 & \Cdots & \fcand_{3\setcovsize} & \cand_{3\setcovsize} \\
   1 & 4\setcovsize + 1 & 7\setcovsize -1 & & & & & & & & \\
   2 & 7\setcovsize - 1 & 4\setcovsize + 1 & & & & & & & & \\
   3 & & & 11\setcovsize & & & & & & & \\
   \Vdots & & & & \Ddots & & & & & & \\
   & & & & & & & & & & \\
   & & & & & & & & & & \\ 
   & & & & & & & &  &  & \\ 
   & & &  & & & & & & & \\
   & &  & & & & & & & & \\
   9\setcovsize +2 & &  & & & & & & & &11\setcovsize \\
  \end{bNiceMatrix}
  \end{small}
	$$
 }
 \newcommand{\helperMatrixDFullExpr}{%
 $$
  B \coloneqq
  \begin{small}
  \setlength{\arraycolsep}{3pt}
  \begin{bNiceMatrix}[first-row, nullify-dots, code-for-first-row =
  \scriptstyle, code-for-first-col = \scriptstyle, first-col]
   & e_1 & e_2 & \Cdots & e_i & \Cdots & e_{3\setcovsize -1} & e_{3\setcovsize}\\
   6\setcovsize+4 & 9\setcovsize - 3 & 3 & & & & & \\
   6\setcovsize+5 & & 9\setcovsize -6 & & & & & \\
   \Vdots & & & \Ddots & & & & \\
   6\setcovsize+2+i & & & & 3i-3& & & \\
   6\setcovsize+3+i & & & & 9\setcovsize - 3i & & & \\
   \Vdots & & & & & \Ddots & & \\
   9\setcovsize+1 & & & & & & 9\setcovsize - 6 & \\
   9\setcovsize+2 & & & & & & 3 & 9\setcovsize - 3
  \end{bNiceMatrix}
  \end{small}
	$$
 }
}
{                                                                                
 \newcommand{\matrixOfConstrainsFullExpr}{%
	 \par\noindent%
   \bgroup
	 \centering
	 \includegraphics[width=0.9\linewidth]{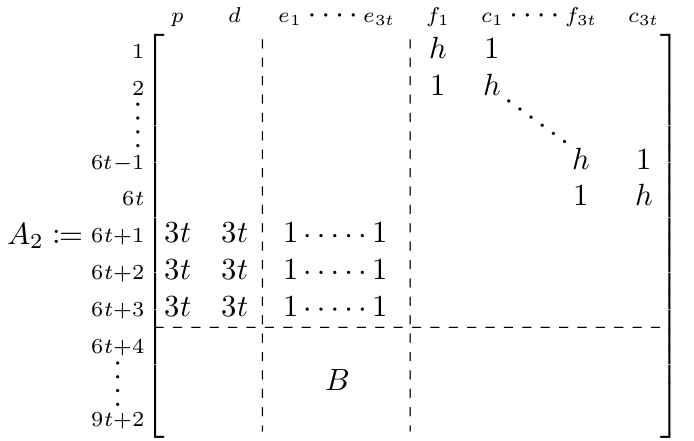}
	
	 \egroup
 }
 \newcommand{\helperMatrixDFullExpr}{%
	 \par\noindent%
   \bgroup
	 \centering
	 \includegraphics[width=0.9\linewidth]{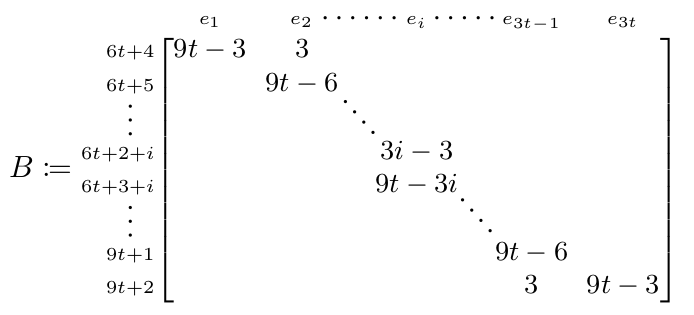}
	
	 \egroup
 }
 \newcommand{\advantageMatrixFullExpr}{
	 \par\noindent%
   \bgroup
	 \centering
	 \includegraphics[width=0.9\linewidth]{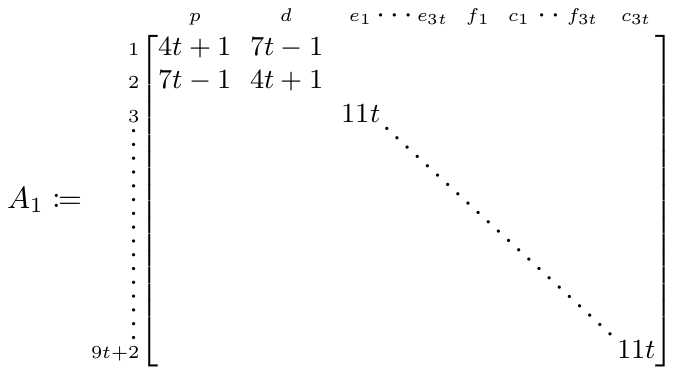}
	
	 \egroup
 }
}                                                                                
\makeatother                                                                     

\newcommand{\holantProb}[1]{\textsc{{#1}-Holant}}

\DeclareMathOperator{\Hol}{Holant}

\newtheorem{proposition}[theorem]{Proposition}
\newtheorem{definition}{Definition}
\newtheorem{example}{Example}
\newtheorem{claim}{Claim}
\usepackage{graphicx}
\usepackage{natbib}
\usepackage{booktabs}
\usepackage{subcaption}
\usepackage{pgfplots}
 \usepackage[overload]{empheq}

\usepackage{mathrsfs}

\newcommand{\np}{{{\mathrm{NP}}}}

\newcommand{\sharpp}{\ensuremath{{\mathrm{\#P}}}}

\newcommand{\vtrs}{\ensuremath{\mathcal{V}}}
\newcommand{\vtrscnt}{\ensuremath{n}}
\newcommand{\voter}{\ensuremath{v}}
\newcommand{\cnds}{\ensuremath{\mathcal{C}}}
\newcommand{\cndscnt}{\ensuremath{m}}
\newcommand{\cand}{\ensuremath{c}}
\newcommand{\elct}{\ensuremath{\mathcal{E}}}
\newcommand{\flct}{\ensuremath{\mathcal{F}}}
\newcommand{\pos}{{\mathrm{pos}}}
\DeclareMathOperator{\swap}{swap}
\newcommand{\pref}{\succ}
\newcommand{\emd}{\mathrm{emd}}
\newcommand{\calD}{\mathcal{D}}
\newcommand{\calU}{\mathcal{U}}
\newcommand{\calS}{\mathcal{S}}

\newcommand{\pmatrofall}[1]{\ensuremath{P}({#1})}

\newcommand{\realizationsProb}{\textsc{\#Realizations}}
\newcommand{\lexrealizationsProb}{\textsc{\#LexRealizations}}
\newcommand{\calT}{{{\mathcal{T}}}}

\newcommand{\omitappendix}[2]{#1}   

\title{Properties of Position Matrices and Their Elections}
\author {
    Niclas Boehmer,\textsuperscript{\rm 1}
    Jin-Yi Cai,\textsuperscript{\rm 2}
    Piotr Faliszewski,\textsuperscript{\rm 3}
    Austen Z. Fan,\textsuperscript{\rm 2}
    Łukasz Janeczko,\textsuperscript{\rm 3}\\
    Andrzej Kaczmarczyk,\textsuperscript{\rm 3}
    Tomasz Wąs\textsuperscript{\rm 3, \rm 4}\\
    \\
    \textsuperscript{\rm 1} Algorithmics and Computational Complexity, Technische Universität Berlin\\
    \textsuperscript{\rm 2} University of Wisconsin-Madison\\
    \textsuperscript{\rm 3} AGH University\\
    \textsuperscript{\rm 4} Pennsylvania State University\\
    {\small niclas.boehmer@tu-berlin.de, jyc@cs.wisc.edu, faliszew@agh.edu.pl, afan@cs.wisc.edu, ljaneczk@agh.edu.pl,}\\
    {\small andrzej.kaczmarczyk@agh.edu.pl, twas@psu.edu}
}
\date{March 3, 2023}

\begin{document}
\maketitle

\begin{abstract}
  We study the properties of elections that have a given position
  matrix (in such elections each candidate is ranked on each position
  by a number of voters specified in the matrix).  We show that counting elections that
  generate a given position matrix is
  $\sharpp$-complete. Consequently, sampling such elections uniformly
  at random seems challenging and we propose a simpler algorithm,
  without hard guarantees. Next, we consider the problem of testing if
  a given matrix can be implemented by an election with a certain
  structure (such as single-peakedness or
  group-separability). Finally, we consider the problem of checking if
  a given position matrix can be implemented by an election with a
  Condorcet winner.  We complement our theoretical findings with
  experiments.
\end{abstract}

\section{Introduction}

Studies of voting and elections are at the core of computational
social choice~\citep{bra-con-end-lan-pro:b:comsoc-handbook}.
An (ordinal) election is represented by a set of candidates and a
collection of voters
who rank the candidates from the most to the least appealing one. Such
preferences are sometimes shown in an aggregate form as a
\emph{position matrix}, which specifies for each candidate the number
of voters that rank him or her on each possible position. 
Motivated by the connection of position matrices  to the so-called maps
of elections, and their similarity to weighted majority relations,
 we study the properties of elections with a given position
matrix.
%


The idea of a map of elections, introduced by
\citet{szu-fal-sko-sli-tal:c:map} and
\citet{boe-bre-fal-nie-szu:c:compass}, is to collect a set of
elections, compute the 
distances between them, and embed the elections as points in the plane,
so that the Euclidean distance between points resembles the distance between the respective elections.  Such maps are
useful because nearby elections seem to have similar properties (such
as, e.g., running times of winner determination algorithms, scores of
winning candidates, etc.; see, e.g., the works
of~\citet{szu-fal-sko-sli-tal:c:map},
\citet{boe-bre-fal-nie:c:counting-bribery}, and
\citet{boe-sch:t:datasets}).
However, there is a catch. The positionwise distance, which is
commonly used in these maps, views elections with the same
position matrix as identical.  Hence there might exist
very different elections that, nonetheless, have
identical position matrices and in a map 
are placed on
top of each other. We want to evaluate to what extent this issue
constitutes a problem for maps of elections.

The second motivation for our studies is that
position matrices are natural counterparts of weighted majority
relations, which specify for each pair of candidates how many voters
prefer one to the other. While weighted majority relations provide
sufficient information to determine winners of many
Condorcet-consistent voting rules,\footnote{Rules that can be computed using
only the weighted majority relation are called C2
  by~\citet{fis:j:condorcet}; see also the overview
  of \citet{zwi:b:intro-voting}. A Condorcet winner is
  preferred to every other candidate by a majority of
  voters. Condorcet-consistent rules always select Condorcet winners
  when they exist.
  Some non-Condorcet-consistent rules are also C2
  (e.g., the Borda rule).}
  position matrices provide
the information needed by positional scoring rules (i.e., rules where
each voter gives each candidate a number of points that depends on
this candidate's position in his or her ranking).  Together with
Condorcet-consistent rules, positional scoring rules are among the
most widely studied single-winner voting rules.
While weighted majority relations are commonly studied and analyzed
(even as early as in the classic theorem of
\citet{mcg:j:election-graph}), position matrices have not been studied
as carefully.

Our contributions regard three main
issues. 
First, we ask how similar are elections that have the same position
matrix. To this end, we would like to sample elections with a given
position matrix uniformly at random. Unfortunately, doing so appears
to be challenging.
In particular, a natural sampling algorithm requires the ability to
count elections that generate a given position matrix, and we show
that doing so is $\sharpp$-complete. While, formally, there may exist
a different approach, perhaps providing only an approximately uniform
distribution, finding it is likely to require significant effort
(indeed, researchers have been trying to solve related sampling problems for
quite a while, without final success as of now; see, e.g., the works
of~\citet{jac-mat:j:generating-latin} and
\citet{hon-mik:t:sampling-edge-colorings}).
%
%
%
%
%
We design a simpler sampling algorithm,
without hard guarantees on the distribution, and use it to evaluate how
different two elections with a given position matrix can be. The algorithm,
albeit not central to our study, might be of
independent interest
when considering sampling various preference
distributions~\citep{reg-gro-mar-tse:behav-social-choice,tid-pla:b:modeling-elections,all-gol-jus-mat:uni-ran-gen-cp-nets}.



Second, we consider structural properties of elections that generate a
given position matrix (or its normalized variant, called a frequency
matrix). Specifically, given a 
matrix we ask if there is an election that generates it and whose
votes come from a given domain (such as the single-peaked
domain~\citep{bla:b:polsci:committees-elections}, some group-separable
domains~\citep{ina:j:group-separable,ina:j:simple-majority}, or a
domain given explicitly vote-by-vote as part of the input).  We show
polynomial-time algorithms that, given a frequency matrix and a
description of
a domain (e.g., via a single-peaked axis or by listing the votes
explicitly), decides if there is an election with votes from this domain
that generates this matrix.  We apply these algorithms to test which
frequency matrices from the map of elections can be generated from
elections with a particular structure.\footnote{We form a map that is
  analogous to that used by \citet{boe-bre-fal-nie-szu:c:compass}, but
  which uses 8 candidates rather than 10 (using fewer candidates helps
  significantly with our computation times).}

Finally, we consider the problem of deciding for a given position
matrix if there is an election that implements the matrix and has a
Condorcet winner (i.e., a candidate who is preferred to every other
one by a strict majority of voters). We evaluate experimentally which
matrices from our map have such elections, provide a necessary
condition for such elections to exist, and check how often this
condition is effective on the map of elections.  Additionally, for
each matrix from the map we compute for how many different candidates
there is an election that generates this matrix and where this candidate
is a Condorcet winner. 


With our theoretical and empirical analysis, we ultimately want to
answer the question how much information is contained in a position
matrix and how much flexibility is still left when implementing
it.%
\footnote{
The code for the experiments is available at:
\url{https://github.com/Project-PRAGMA/Position-Matrices-AAAI-2023}.}

\section{Preliminaries}\label{sec:prelim}
For each $k \in \mathbb{N}_+$, by $[k]$ we denote the set
$\{1,\dots,k\}$.  Given a matrix $X$, by~$X_{i,j}$ we mean
its entry in row~$i$ and column~$j$.  For two equal-sized sets $X$ and $Y$,
by $\Pi(X,Y)$ we mean the set of one-to-one mappings from $X$ to
$Y$. $S_n$ is a shorthand for $\Pi([n],[n])$, i.e., the
 set of permutations of $[n]$.

An \emph{election}~$\elct{}$ is a pair $(\cnds, \vtrs)$ consisting of
a set~$\cnds{}=\{\cand_1, \cand_2, \ldots, \cand_\cndscnt\}$ of
\emph{candidates} and a
collection~$\vtrs{} = (\voter_1, \voter_2, \ldots, \voter_\vtrscnt)$
of \emph{votes}, i.e., complete, strict orders over the
candidates. These orders rank the candidates from the most to the
least appealing one according to a given voter (we use the terms
``vote'' and ``voter'' interchangeably).
If some
voter~$\voter$ prefers candidate~$\cand$ over candidate $\cand'$,
then we write $\cand \succ_{\voter} \cand'$; we omit the subscript when it is
clear from context.  Given a vote
$v_i \colon c_1 \succ c_2 \succ \cdots \succ c_m$, we say that $v_i$
ranks $c_1$ on the first position, $c_2$ on the second one, and
so on.
For two votes $u$ and $v$ over the same candidate set, their swap
distance, $\swap(u,v)$, is the smallest number of swaps of adjacent
candidates necessary to transform $u$ into $v$.

In an election~$\elct{} = (\cnds, \vtrs)$, a
candidate~$\cand \in \cnds$ is a~\emph{Condorcet winner} of
the election if for every other candidate~$d$ more than half of the
voters prefer~$\cand$ to~$d$.

\subsection{Position and Frequency Matrices}

Let $\elct$ be some election and assume that the candidates are
ordered in some way
(e.g., lexicographically, with respect to their names).
The \emph{position matrix} of $\elct$ (with respect to this
order) is a non-negative, integral $\cndscnt \times \cndscnt$ matrix
$X$ such that for each $i,j \in [\cndscnt]$, $X_{i,j}$ is the number
of voters that rank the $j$-th candidate on the $i$-th position.  By
$\pmatrofall{\elct}$ we denote the set of all position matrices
of~$\elct$ for all possible orderings of candidates. Note that the
matrices in $\pmatrofall{\elct}$ only differ by the order of their
columns. 

%

For a position matrix $X \in \pmatrofall{\elct}$, where $\elct$ is an
election with $n$ voters, the corresponding \emph{frequency matrix} is
$Y := \frac{1}{n}\cdot X$. In other words, frequency matrices are
normalized variants of the position ones, where each value $Y_{i,j}$
gives the fraction of voters that rank the $j$-th candidate on the
$i$-th position.
Every frequency matrix is \emph{bistochastic}, i.e., the elements in
each row and in each column sum up to one.  Hence, 
we often refer to bistochastic matrices as frequency
matrices, and to integral square matrices with nonnegative entries,
where each row and each column sums up to the same value, as
position matrices.


We say that an election~$\elct$ \emph{realizes} (or \emph{generates}) a position matrix~$X$ (or, a frequency matrix $Y$)
if~$X \in \pmatrofall{\elct}$ (or, $n \cdot Y \in \pmatrofall{\elct}$,
where $n$ is the number of voters in $\elct$).
\citet{boe-bre-fal-nie-szu:c:compass} showed that every
position matrix $X$ 
is realizable by some
election 
(their result is a
reinterpretation of an older result of
\citet{mee-mye:j:birkhoff-neumann-etc}). 
\citet{DBLP:journals/tcs/YangG16} also showed that position matrices are always realizable as part of a proof that they can be used to solve a Borda manipulation problem.
Note
that 
two distinct elections may generate the same position matrix.


\begin{example}\label{ex:posmat}
  Consider an election $\elct$ with candidates $a$, $b$, $c$, and $d$
  and four votes shown below on the left. On the right we show a
  position matrix of this election (for the natural ordering of the
  candidates):

  {\centering \vspace{-0.25cm}
    \begin{minipage}[b]{0.3\columnwidth}
      \small
      \begin{align*}
        &v_1 \colon a \pref b \pref c \pref d, \\
        &v_2 \colon b \pref a \pref d \pref c, \\
        &v_3 \colon a \pref b \pref d \pref c, \\
        &v_4 \colon b \pref a \pref c \pref d. 
      \end{align*}
    \end{minipage}\quad\quad
    \begin{minipage}[b]{0.4\columnwidth}
      \small
      \begin{align*}
        \kbordermatrix{ & a & b & c & d  \\
        1 &                2 & 2 & 0 &  0\\
        2 &                2 & 2 & 0 &  0\\
        3 &                0 & 0 & 2 &  2\\
        4 &                0 & 0 & 2 &  2\\
   }
      \end{align*}
    \end{minipage}\\[1mm]
  }

  \noindent Note that this is also a position matrix for an election
  with two votes $a \pref b \pref c \pref d$ and two votes
  $b \pref a \pref d \pref c$.
\end{example}

\subsection{Structured Domains}
We are 
interested in elections where the votes have some
structure.  For example, the single-peaked domain
captures votes on the political left-to-right spectrum (and, more generally,
votes focused on a single issue, such as those regarding the
temperature in a room or the level of taxation).
\begin{definition}
  An election~$\elct{} = (\cnds, \vtrs)$ is \emph{single-peaked} if
  there is an order~$\triangleright$ (the \emph{societal axis}) over
  candidates~$\cnds$ such that for each vote~$\voter \in \vtrs$ and
  for each $\ell \leq |\cnds|$, the top~$\ell$~candidates according
  to~$\voter$ form an interval with respect to~$\triangleright$.
\end{definition}
Intuitively, in a single-peaked election each voter first selects their
favorite candidate and, then, extends his or her ranking step by
step with either the candidate directly to the left or directly to the
right (wrt.~$\triangleright$)  of those already ranked.

Group-separability captures settings where the candidates have some
features and the voters have hierarchical preferences over these
features. Let $\cnds$ be a set of candidates and consider a rooted,
ordered tree~$\calT$, where each leaf one-to-one corresponds to a
candidate. A \emph{frontier} of $\calT$ is a vote that we obtain by
reading the names of the candidates associated with the leaves
of~$\calT$ from left to right. A vote is \emph{compatible} with
$\calT$ if it can be obtained as its frontier by reversing for some
nodes in $\calT$ the order of their children.  Intuitively, we view the
internal nodes of $\calT$ as features and a candidate has the
features that appear on the path from it to the root.

\begin{definition} An election $\elct = (\cnds,\vtrs)$ is
  \emph{group-separable} if and only if there is a tree $\calT$ over
  candidate set $\cnds$ such that each vote from $\vtrs$ is compatible
  with $\calT$.
\end{definition}

We focus on 
balanced trees (i.e., complete binary trees)
and on caterpillar trees (i.e., binary trees where each non-leaf has
at least one leaf as a child).
If an election is group-separable 
for 
a
balanced tree, then we say that this elections is \emph{balanced
  group-separable}. Analogously, we speak of \emph{caterpillar
  group-separable elections}.

\begin{example}
  The election from Example~\ref{ex:posmat} is both single-peaked (for
  societal axis
  $c \mathrel\triangleright a \mathrel\triangleright b
  \mathrel\triangleright d$) and balanced group-separable (for a tree
   whose frontier is $a \pref b \pref c \pref d$).
\end{example}

Single-peaked elections were introduced by
\citet{bla:b:polsci:committees-elections}, and group-separable ones by
\citet{ina:j:group-separable,ina:j:simple-majority}. We mention that
Inada's original definition is different from the one that we
provided, but they are equivalent~\citep{kar:j:group-separable} and
the tree-based one is algorithmically much more convenient.  We point
readers interested in structured domains to the recent survey of
\citet{elk-lac-pet:t:restricted-domains-survey}.

\begin{figure}[t]
\centering
  \includegraphics[width=7cm]{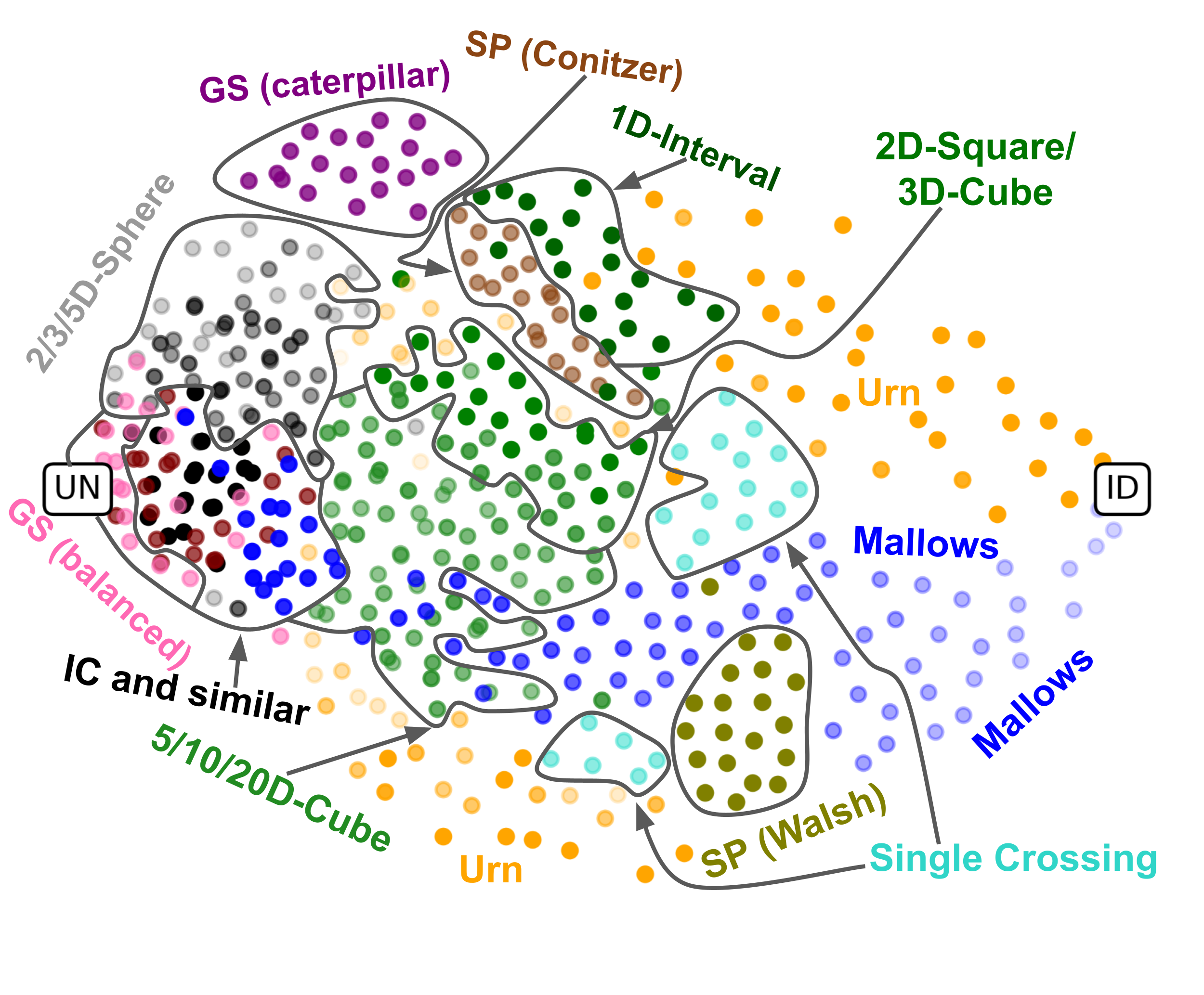}
  \caption{Map of elections visualizing the 8x80~dataset. Each dot
    represents an election and its color corresponds to the
    statistical model used to generate it. $x$D-Cube/Sphere models are
    Euclidean models where the points of the candidates and voters are
    chosen uniformly at random from an $x$-dimensional
    hypercube/hypersphere ($1$D-Interval is $1$D-Cube;
    $2$D-Square is $2$D-Cube). For Mallows and Urn elections the
    transparency of the coloring indicates the value of the used
    parameter. \omitappendix{For the other models, see \cref{app:map}.}{}}\label{fig:main-map}
\end{figure}

\subsection{Map of Elections}\label{sec:map-of-elections}

For our experiments, we use an \emph{8x80}~dataset that resembles those of
\citet{szu-fal-sko-sli-tal:c:map},
\citet{boe-bre-fal-nie-szu:c:compass}, and
\citet{boe-fal-nie-szu-was:c:metrics}. 
It
contains $480$~elections 
with $8$~candidates and $80$~votes
generated using the same statistical models, with the same parameters,
as the map of \citet{boe-fal-nie-szu-was:c:metrics}. In particular, we used
\begin{enumerate*}[label=(\roman*)]
  \item impartial culture (IC), where each vote is equally likely,
  \item the Mallows and urn  distributions, whose votes are more or less correlated,
depending on a parameter,
\item various Euclidean models, where candidates and voters are points in
Euclidean spaces and the voters rank the candidates with respect to
their geometric distance, and
\item uniform distributions over balanced
group-separable, caterpillar group-separable, and single-peaked
elections (we refer to the uniform distribution of single-peaked
elections as the Walsh model; we also use the model of
\citet{con:j:eliciting-singlepeaked} to generate single-peaked
elections).
\end{enumerate*}
%
%
\omitappendix{See \cref{app:map} for exact descriptions.}{See the full paper for exact descriptions.}
We repeated all our experiments from Sections~\ref{sec:structure} and~\ref{sec:condorcet} on analogously composed datasets with a varying number of candidates and voters.
Specifically, we considered elections with either $4$ or $8$
candidates and either $40$, $80$, or $160$ voters. The results on those datasets were similar to those for the 8x80 one.

We present our dataset as a map
of elections, i.e., as points on a plane, where each point corresponds
to an election (see \cref{fig:main-map}).
The Euclidean distances between the points resemble
positionwise distances between the respective elections. For a
definition of the positionwise distance, we point the reader to the
work of \citet{szu-fal-sko-sli-tal:c:map} or to \omitappendix{\cref{app:map}}{the full version}; an
important aspect of this distance is that for two elections~$\elct$
and~$\flct$ (with the same numbers of candidates and~voters) it
depends only on $\pmatrofall{\elct}$ and~$\pmatrofall{\flct}$. Hence,
we will also sometimes speak of the distance between position
matrices.

Our maps include two special position matrices, the uniformity one
(UN), which corresponds to elections where each candidate is ranked on
each position equally often, and the identity one (ID), which
corresponds to elections where all votes are identical.  ID models
``perfect order,'' whereas UN models ``perfect chaos'' (but note that there
exist very structured elections whose position matrix is UN). UN
and ID, as well as two other special points, were introduced by
\citet{boe-bre-fal-nie-szu:c:compass}.  For each two elections, their
positionwise distance is at most as large as the distance between UN and
ID~\citep{boe-fal-nie-szu-was:c:metrics}.

\section{Counting and Sampling Elections}\label{sec:counting}


Given a position matrix, it would be useful to be able to sample
elections that realize it uniformly at random. 
Unfortunately, doing so
seems challenging. Indeed, one of the 
natural sampling 
algorithms \omitappendix{(presented in \cref{app:sampler})} requires, among others, the
ability to count elections that realize a given matrix, a task which we show
to be $\sharpp$-complete. While, formally, this 
does not
preclude the existence of a polynomial-time uniform sampler (and,
certainly, it does not preclude the existence of an approximately
uniform one), we believe that it suggests that finding such algorithms
would require deep 
insights; for closely related
problems such insights are still
elusive~\citep{jac-mat:j:generating-latin,hon-mik:t:sampling-edge-colorings}.

Formally, in the \realizationsProb{} problem we are given an
$m \times m$ position matrix $X$ (and a candidate~set~$\{c_1, \ldots, c_m\}$,
where, for each $i$, candidate $c_i$ corresponds to the $i$-th column of $X$)
and we ask for the number of elections that realize $X$. Two elections are
distinct
if their voter collections are distinct when viewed as multisets.




\begin{theorem}\label{thm:count-real-sharp}
  \realizationsProb{} is \#P-complete even if the realizing elections contain three votes.
\end{theorem}

\subsection[Preparing for the Proof of Theorem 1]{Preparing for The Proof of~\cref{thm:count-real-sharp}}
We first provide the necessary background for our proof of
\cref{thm:count-real-sharp}. Given a graph $G$, directed or
undirected, a $t$-edge coloring is a function that associates each of its
edges with one of $t$ colors. Such a coloring is proper if for each
vertex the edges that touch it have different colors. A graph is
$r$-regular if each vertex touches exactly $r$ edges (for directed
graphs, both incoming and outgoing edges count).  The \#P-hardness of \realizationsProb{} follows by a reduction from the problem
of counting proper $3$-edge colorings of a given $3$-regular bipartite
(simple) graph. We refer to this problem as
\textsc{3-Reg.-Bipartite-3-Edge-Coloring}. 
We start by establishing that this  problem is $\sharpp$-hard.
To prove this, we will give a reduction from a specific Holant
problem, which we will call \textsc{Holant-Special}. In this problem
we are given a planar, $4$-regular, directed graph~$G$, where each
vertex has two incoming edges and two outgoing ones. Further, we have
an embedding of this graph on the plane, which has the following
property: As we consider the edges touching a given vertex in the
counter-clockwise order, every other edge is incoming and every other
one is outgoing.  Let $\mathscr{C}$ be the set of all $3$-edge-colorings of
$G$. Given a vertex $v$, its four touching edges $e_1, \ldots, e_4$ (listed in
the counter-clockwise order, starting from some arbitrary one)
and some coloring
$\sigma \in \mathscr{C}$, we denote by $\sigma(v)$ the vector
$(\sigma(e_1), \ldots, \sigma(e_4))$. We define a function $f$ so
that:
\begin{enumerate}
\item $f(\sigma(v)) = 0$ if $\sigma(v)$ includes three different
  colors,
\item $f(\sigma(v)) = 2$ if all colors in $\sigma(v)$ are identical,
\item $f(\sigma(v)) = 1$ if $\sigma(v)$ includes two different colors and
  there are two consecutive edges in the counter-clockwise order that have
  the same color,
\item $f(\sigma(v)) = 0$ otherwise (i.e., if $\sigma(v)$ includes two
  different colors and each two consecutive edges in the
  counter-clockwise order have different colors).
\end{enumerate}
%
The goal is to compute
$\sum_{\sigma \in \mathscr{C}}\prod_{v \in V(G)} f(\sigma(v))$.
\citet{cai-guo-wil:compl-count-colorings} have shown that doing so is
$\sharpp$-complete (their results are far more general than this; the
problem we consider is a variant of their
$\langle 2,1,0,1,0 \rangle$-\textsc{Holant} problem).  The left-hand side
of \cref{fig:holant-pictures} shows an example input for
\textsc{Holant-Special}.

\begin{theorem}\label{lem:bip-three-col}
  \textsc{\#$3$-Reg.-Bipartite-3-Edge-Colo\-ring} is $\sharpp$-hard.
\end{theorem}
\begin{proof}
  We give a reduction from \textsc{Holant-Special} to
  \textsc{\#$3$-Reg.-Bipartite-3-Edge-Colo\-ring}. The construction is
  inspired by those used by
  \citet{cai-guo-wil:compl-count-colorings}. Let $G=(V,E)$ be the input
  graph and let the notation be as in the discussion preceding the
  theorem statement.

  The high-level idea 
  is to modify graph $G$ by replacing each vertex $v\in V$
  with a gadget, while keeping ``copies'' of edges from $E$. 
  Then, the value of~$f(\sigma(v))$ for some edge-coloring~$\sigma$ of the edges from $E$
  in $G$ corresponds to the number of proper $3$-edge-colorings in the gadget
  for $v$ assuming the ``copies'' of $E$ in the
  constructed graph are colored according to~$\sigma$.  
  Specifically, we replace each vertex $v$ by the gadget depicted in the right-hand side
  of~\cref{fig:holant-pictures}.  Its four dangling edges implement
  the original four edges of $v$. However, we need 
  some care in deciding which of the dangling edges we
  connect to which vertices from the gadgets corresponding to the neighbors of $v$ in $G$  (we will return to this issue after we
  explain how the gadget works).

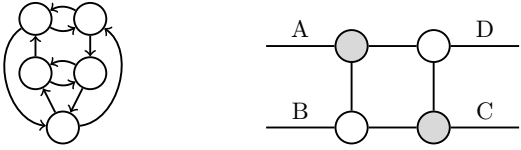
\begin{figure}
 \centering%
 \scalebox{0.9}{
 \begin{subfigure}[b]{0.2\textwidth}
  \begin{tikzpicture}
   \begin{scope}[scale = 0.8, every node/.style = {circle, thick, draw = black, text width = 5pt}]
    \node (A) at (0,0) {};
    \node (B) at (0,1) {};
    \node (C) at (1,1) {};
    \node (D) at (1,0) {};
    \node (E) at (.5,-1) {};
   \end{scope}
   \begin{scope}[every edge/.style = {draw = black, thick, bend right}]
    \draw[->] (A) edge[bend right = 0] (B); 
    \draw[->] (A) edge (D); 
    \draw[->] (B) edge (C); 
    \draw[->] (B) edge[bend right = 70] (E); 
    \draw[->] (C) edge[bend right = 0] (D); 
    \draw[->] (C) edge (B); 
    \draw[->] (D) edge[bend right = 0] (E); 
    \draw[->] (D) edge (A); 
    \draw[->] (E) edge[bend right = 70] (C); 
    \draw[->] (E) edge[bend right = 0] (A); 
   \end{scope}
  \end{tikzpicture}
 \end{subfigure}%
 \hspace{0.3cm}
 \begin{subfigure}[b]{0.2\textwidth}
  \begin{tikzpicture}
   \begin{scope}[scale = 1.2, every node/.style = {circle, thick, draw = black, text width = 5pt}]
    \node (A) at (0,0) {};
    \node[fill = black!15] (B) at (0,1) {};
    \node (C) at (1,1) {};
    \node[fill = black!15] (D) at (1,0) {};
    \coordinate[left = of A] (dA);
    \coordinate[left = of B] (dB);
    \coordinate[right = of C] (dC);
    \coordinate[right = of D] (dD);
   \end{scope}
   \begin{scope}[every edge/.style = {draw = black, thick}]
    \draw (A) edge (B); 
    \draw (A) edge (D); 
    \draw (B) edge (C); 
    \draw (C) edge (D); 
    \draw (dA) edge node[midway, above] {B}(A); 
    \draw (dB) edge node[midway, above] {A}(B); 
    \draw (dC) edge node[midway, above] {D}(C); 
    \draw (dD) edge node[midway, above] {C}(D); 
   \end{scope}
  \end{tikzpicture}
 \end{subfigure}
 }
 \caption{An example input graph (left) and the gadget used in the proof
 of~\cref{lem:bip-three-col}. The letters label the gadget's dangling edges. The
 colors  illustrate its
 bipartiteness.
 \label{fig:holant-pictures}}
\end{figure}

  For each of our gadgets, we name the dangling edges $A$, $B$, $C$,
  and $D$, as shown in~\cref{fig:holant-pictures}.
  It is now easy to
  see that if all dangling edges are of the same color, say $1$, then
  there are two colorings of the remaining edges of the gadget resulting in a proper coloring:
  Both ``vertical'' edges need to have the same color (either $2$ or
  $3$), and both ``horizontal'' edges need to have the same color (the
  single remaining one).
  Similarly, if edges $A$ and $B$ have the same color, and edges $C$
  and $D$ have the same color, then there is a unique proper coloring of
  the other edges. By symmetry, the same holds if both edges $A$ and $D$
  and edges $B$ and $C$ have the same color.
  Finally, if edges $A$ and $C$ have the same color, and edges $B$ and
  $D$ have the same color (or, the dangling edges have three different
  colors) then there are no proper colorings of the remaining edges in
  the gadget.
  This way, for each vertex $v$ and coloring~$\sigma$, $v$'s gadget implements
  the $f(\sigma(v))$ function.

  Next we describe how we connect the dangling edges of the gadgets.
  If $u$ and $v$ are two vertices of $G$ and there is a directed edge
  from $u$ to $v$, then we merge one of the $A$ and~$C$ dangling edges
  of $u$'s gadget with one of the $B$ and~$D$ dangling edges of $v$'s
  gadget (which dangling edges we use is irrelevant for this proof).
  Since each vertex in $G$ has two incoming and two outgoing edges,
  doing so is possible.
  
  As the gadgets are bipartite themselves, and due to the way in
  which we connect their edges, the resulting graph~$G'$ is
  bipartite. It is also clear that it is $3$-regular. Finally, due to
  the way in which $3$-edge-colorings of $G$ can be extended to proper
  $3$-edge-colorings of $G'$ (see the description of the gadgets), we
  see that the number of the latter is equal to the output of
  the \textsc{Holant-Special} for $G$.  The reduction runs in
  polynomial-time and the proof is complete.
\end{proof}

The above result also applies to $3$-regular planar bipartite
graphs. To see this, it suffices to appropriately arrange our gadgets
in space (sometimes rotating them) and choose the dangling edges to
connect more carefully.

\subsection[The Proof of Theroem 1]{The Proof of \cref{thm:count-real-sharp}}

The answer to an instance of~\realizationsProb{} is the number of
accepting paths of a non-deterministic Turing machine that constructs an
election and then checks if it realizes the input 
matrix. As this
machine works in (non-deterministic) polynomial time, \realizationsProb{} is
in~$\sharpp$. 

To show \sharpp{}-hardness,
we give a reduction from \textsc{\#$3$-Reg.-Bipartite-3-Edge-Coloring}
to~\realizationsProb{}.
Let $G = (U,V; E)$ be our input 3-regular bipartite graph, where $U$
is the set of vertices on the left, $V$ is the set of vertices on the
right, and $E$ is a set of edges. Since $G$ is $3$-regular, we have
$|U| = |V|$. W.l.o.g., we let $U = \{u_1, \ldots, u_m\}$ and
$V = \{v_1, \ldots, v_m\}$.
We form
an $m \times m$ matrix~$X$, where each entry $X_{i,j}$ is either $1$,
if there is an edge between $v_i$ and $u_j$, or $0$, if there is no
such edge.
%
As $G$ is $3$-regular, $X$ has exactly three ones in each row and in
each column, so it is a position matrix and each election that
realizes it contains three votes.

We now show that each proper $3$-edge-coloring of~$G$ gives an election
realizing matrix~$X$. For a given coloring, the edges of the same
color form a perfect matching in~$G$. We interpret such a matching as
a single vote. Specifically, we treat vertices from~$U$ as candidates
and vertices from~$V$ as positions in the vote being constructed
(e.g., if the matching contains an edge between $v_i$ and $u_j$, then
the vote ranks candidate $u_j$ on position $i$). Hence, for each
$3$-coloring we get an election consisting of three votes, one for
each matching associated with one color.  Since all edges must be part
of some matching and each edge corresponds to a single $1$-entry
in~$X$, the resulting election realizes~$X$.

Each election realizing matrix~$X$ corresponds to
six  
$3$-edge-colorings of~$G$. Indeed, taking one~$3$-edge-coloring,
each of the six~permutations of the colors gives raise to the same
election. This holds, because for a single $3$-edge-coloring, each color
forms an edge-disjoint matching (as opposed to graphs with parallel
edges, where this would not be true). So our reduction
preserves the number of solutions with a multiplicative 
factor of~$6$. This completes the proof.

\subsection{Experiments}\label{sec:distance-exp}

We checked experimentally  how diverse are elections
that generate the same position matrix. To do so, we used
the isomorphic swap distance, due to
\citet{fal-sko-sli-szu-tal:c:isomorphism}.
\begin{definition}
  Let $\elct = (\cnds,\vtrs)$ and $\flct = (\mathcal{D}, \mathcal{U})$
  be two elections, where $\cnds = \{c_1, \ldots, c_m\}$,
  $\mathcal{D} = \{d_1, \ldots, d_m\}$, $\vtrs = (v_1, \ldots, v_n)$,
  and $\mathcal{U} = (u_1, \ldots, u_n)$. Their isomorphic swap
  distance is:
  \[
    d_{\swap}(\elct, \flct) = \min_{\sigma \in S_n}\min_{\pi \in \Pi(\cnds, \mathcal{D})}
    \textstyle\sum_{i=1}^n \swap( \pi(v_i), u_{\sigma(i)} ),
  \]
  where $\pi(v_i)$ is the vote $v_i$ where each candidate
  $c \in \cnds$ is replaced with candidate $\pi(c)$.
\end{definition}
Intuitively, the isomorphic swap distance between two elections is the
summed swap distance of their votes, provided we first rename the
candidates and reorder the votes to minimize this value. Maps of
elections could be generated using the isomorphic swap distance
instead of the positionwise one, and they would be more accurate than
those based on the positionwise
distance~\citep{boe-fal-nie-szu-was:c:metrics}, but the isomorphic swap
distance is $\np$-hard to compute and challenging to compute in
practice~\citep{fal-sko-sli-szu-tal:c:isomorphism}; indeed, we use a
brute-force implementation.

\citet{boe-fal-nie-szu-was:c:metrics} have shown that the largest
isomorphic swap distance between two elections with $m$ candidates and
$n$ voters is $\frac{1}{4}n(m^2-m)$ (up to minor rounding errors; for
this result, see their technical report).  Whenever we give an
isomorphic swap distance between two elections (with the same numbers
of candidates and voters), we report it as a fraction of this value.



As we do not have a fast procedure for sampling
(approximately) uniformly at random elections that realize a given
matrix, we use the following naive approach (let $X$ be an $m \times m$
position matrix):
\begin{enumerate}
\item We form an election $\elct = (\cnds, \!\vtrs)$, where
  $\cnds = (\cand_1, \ldots, \cand_m)$ and $\vtrs$ is initially empty.
  For each $i \in [m]$, candidate $\cand_i$
  corresponds to the $i$-th column of the matrix.

\item We repeat the following until $X$ consists of zeros only: We
  form a bipartite graph with vertices $\cand_1, \ldots, \cand_m$ on
  the left and vertices $1, \ldots, m$ on the right; there is an edge
  between $c_j$ and $i$ exactly if $X_{i,j} > 0$. We draw uniformly at
  random a perfect matching in this graph (it always
  exists;~\cite{mee-mye:j:birkhoff-neumann-etc})---we generate it relying on
   the standard self-reducibility of computing perfect matchings and using
  the classic reduction to computing the
  permanent~\citep{val:j:permanent}, which we compute using the formula of \citet{ryser:comb-mathematics}.\footnote{We used a python module called \textit{permanent}
  (\url{https://git.peteshadbolt.co.uk/pete/permanent}) by~Pete
  Shadbolt. In principle, we could have used an approximately uniform
    sampler that runs in polynomial
    time~\citep{jer-sin-vig:j:sampling-matchings,bez-ste-vaz-vig:j:accelerating-annealing},
    but they are too slow in practice.} Given such a matching, we
  form a vote $v$ where each candidate $\cand_j \in \cnds$ is ranked
  on the position to which he or she is matched. We extend $\elct$
  with vote $v$ and we subtract from $X$ the position matrix of the
  election that contains $v$ as the only vote.
\end{enumerate}
In essence, the above procedure is a randomized variant of the
algorithm presented by \citet{boe-bre-fal-nie-szu:c:compass} to show
that every position matrix is realized by some election.

\begin{figure}[t]
\centering
  \includegraphics[width=6cm]{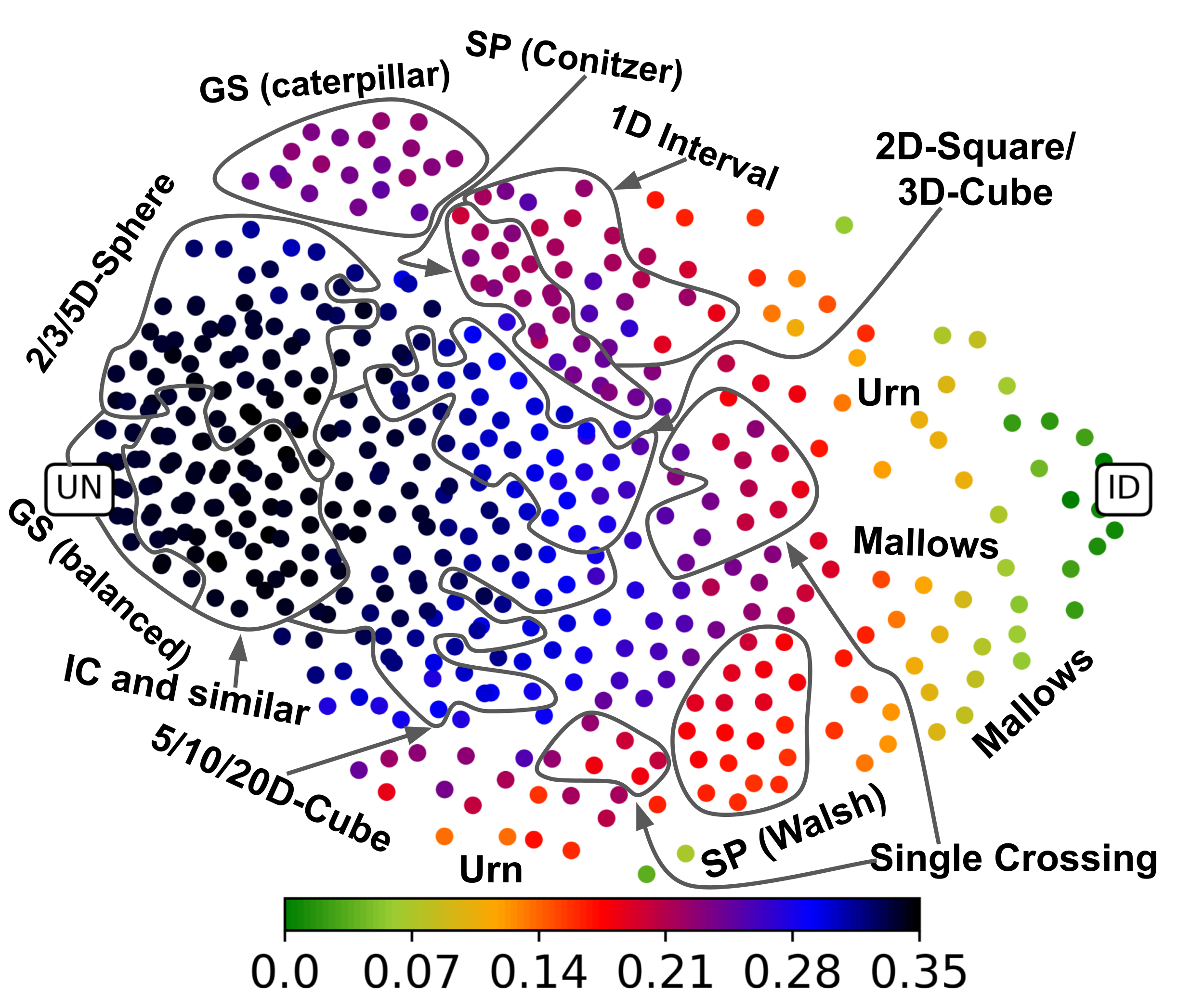}
  \caption{The maximum isomorphic swap distance found for elections realizing a
  given matrix in our 8x80 dataset.}
  \label{fig:maxswap-map}
\end{figure}

We performed the following experiment:
(i)~For each
election from the 8x80~dataset we computed its position matrix, 
(ii)~using
the naive sampler, we generated $100$ pairs of elections that realize it,
and, 
(iii)~for each pair of elections, we computed their isomorphic swap distance. We report the results in Figure~\ref{fig:maxswap-map}, where each dot has
a color that corresponds to the farthest distance computed for the
respective matrix.\footnote{We tried 100 pairs for two reasons. First, each computation is quite expensive. Second, even with testing $10$ pairs the results were very
  similar to those for $100$ pairs
  (if we reported average distances, the results also would not
  change very much).}
For elections close to UN, this distance can be very large.  Indeed, for
about half of the elections (all located close to UN) this
distance is larger than $20\%$ of the maximum possible isomorphic swap
distance. On the other hand, elections realizing position matrices in
the vicinity of ID are much more similar to each other, which is
quite natural.


While we used a naive sampling algorithm rather than a uniform one,
the results are
sufficient to claim that for many position matrices---in particular, those closer to 
UN than to ID---there are two
elections that generate them, whose isomorphic swap distance is very large.
If we had a uniform sampler, the
distances could possibly increase, but the overall conclusion
would not change.
Indeed, we ran our experiment for an analogous
dataset, but for elections with 4 candidates and 16 voters; in this
case we computed maximum possible isomorphic swap distances by generating all elections
realizing a given matrix. 
\omitappendix{The results, presented in~\cref{sub:416exmperiments}, are analogous.}{The results are analogous.}
(For this experiment
we also counted how many elections realize a given matrix and the results
were strongly correlated with the above described results for the maximum distance.)


\section{Recognizing Structure}
\label{sec:structure}

In this section we consider the problem of deciding if a given (arbitrary)
position or frequency matrix can be realized by elections whose votes
come from some domain (e.g., the single-peaked or group-separable one).
Overall, we find that if a precise description of the domain is part
of the input (e.g., if we are given the societal axis for the
single-peaked domain), then for frequency matrices we can often solve
this problem in polynomial time. For position matrices our results are
less positive and less comprehensive.
The reason why frequency matrices are easier to work with here is that
they only specify fractions of votes where a given candidate appears
on a given position, whereas position matrices specify absolute
numbers of such votes and thus are less flexible.


Let us fix a candidate set $\cnds$. We consider sets~$\calD$ of votes,
called domains,
specified in one of the following ways:
\begin{enumerate}
\item \emph{explicit}:
    $\calD$ contains explicitly listed votes, or
\item \emph{single-peaked}:
    $\calD$ contains all votes that are single-peaked
    with respect to an explicitly given axis $\triangleright$, or
\item \emph{group-separable}:
    $\calD$ contains all group-separable votes
    that are compatible with a given rooted, ordered tree $\calT$,
    where each leaf is associated with a unique candidate.
    We only consider balanced or caterpillar trees.
\end{enumerate}
The next theorem is our main result of this section.


\begin{restatable}{theorem}{frecrec}
\label{thm:frecrec}
  There is a polynomial-time algorithm that given a frequency matrix
  $X$ and an explicit, single-peaked,
  or group-separable (balanced or caterpillar) domain $\calD$,
  decides
  if there is an election that realizes $X$,
  and whose votes all belong to $\calD$.
\end{restatable}
%
For example, given a frequency matrix $X$ and a societal axis
$\triangleright$, we can check if there is an election that realizes
$X$ and is single-peaked with respect to $\triangleright$.  A similar
result for single-peakedness and a variant of weighted majority
relations is provided by \citet{spa-wen:c:net-single-peaked}.

\omitappendix{The proof of \cref{thm:frecrec} is quite involved and is relegated
to~\cref{sec:proof_freq_matrices}, but we mention two issues.}{The proof of \cref{thm:frecrec} is quite involved
and is available in the full version of the paper, but we mention two issues.}
First, some of
our algorithms proceed by solving appropriate linear programs and, in
principle, the elections that they discover might have exponentially
many votes with respect to the length of the encoding of the
input. This is not a problem as our algorithms do not build these
elections explicitly.
%
Second, while our algorithms need an explicit description of the domain,
such as the societal axis or the underlying tree, for the
balanced group-separable domain we can drop this assumption, and we can even
deal with position matrices:

\begin{restatable}{theorem}{balanced}
\label{thm:balanced}
  There is a polynomial-time algorithm that given a frequency (or position) matrix
  $X$ decides if the matrix can be
  realized by a balanced group-separable election.
\end{restatable}

Interestingly, if instead of taking the entire balanced group-separable
domain (for a given tree) we only allow for an explicitly specified  subset of its votes, the
problem becomes $\np$-hard.

\begin{theorem}\label{thm:explicit}
  Given a set $\calD$ of votes, listed explicitly, and a position matrix $X$, it is
  $\np$-hard to decide if there is an election that realizes $X$ and
  whose votes are all from $\calD$. This holds even if the votes
  in~$\calD$ are both single-peaked and balanced group-separable.
\end{theorem}

\begin{proof}
  We reduce from the NP-hard X3C problem,
  where we are given a set $\calU = \{u_1, \ldots, u_{3m}\}$ of $3m$
  elements and a family $\calS = \{S_1, \ldots, S_n\}$ of size-$3$
  subsets of $\calU$. We ask if there are~$m$ sets from $\calS$ whose
  union is $\calU$. 
  
  Let $I$ be our input instance of X3C.  We form a $6m\times 6m$
  position matrix $X$ with values $m-1$ on the diagonal, and where for
  each odd column there is value $1$ directly below the $m-1$ entry,
  and for each even column there is value $1$ directly above the $m-1$
  entry (all other entries are equal to~$0$).  The matrix looks as
  follows:
  \[
    \small
    \begin{bmatrix}
    m-1 & 1 & 0 & 0& \cdots & 0\\
    1 & m-1 & 0 & 0& \cdots & 0\\
    0 & 0 & m-1 & 1& \cdots & 0\\
    0 & 0 & 1 & m-1& \cdots & 0\\
    \vdots & \vdots & \vdots & \vdots & \ddots & 1\\
    0 & 0 & 0 & 0 & \cdots & m-1\\
  \end{bmatrix}.
  \]
  We let the candidate set be
  $\cnds = \{\cand_1, \ldots, \cand_{6m}\}$, where for each
  $i \in [6m]$, candidate $\cand_i$ corresponds to the $i$-th
  column. For each set $S_\ell=\{u_i,u_j,u_k\}$ we include in the
  domain $\calD$ a vote~$v_\ell$ that is equal to
  $c_1 \pref c_2 \pref \cdots \pref c_{6m}$ except that $c_{2i}$ and
  $c_{2i+1}$ are swapped, $c_{2j}$ and $c_{2j+1}$ are swapped, and
  $c_{2k}$ and $c_{2k+1}$ are swapped.  We claim that there is an
  election that realizes $X$ and whose votes all belong to $\calD$ if
  and only if $I$ is a \emph{yes}-instance of X3C.

  Let us assume that there are $m$ sets from $\calS$ whose union
  is~$\calU$ and, w.l.o.g., that these sets are $S_1, \ldots,
  S_m$. One can verify that election $(\cnds,(v_1,\ldots, v_m))$
  realizes $X$. Indeed, since $S_1,\ldots, S_m$ cover 
  $\calU$, for each candidate $c_i$ there are $m-1$ votes where~$c_i$ is ranked
  on the $i$-th position, and a single vote where $c_i$ is either ranked one
  position higher or one position lower, depending on the parity of $i$.

  For the other direction, let us assume that there is an election $\elct$ that
  realizes $X$ and, w.l.o.g., that it contains votes
  $v_1, \ldots, v_m$ (all the votes must be distinct as otherwise some
  non-diagonal entry of this election's position matrix would have
  value greater than $1$). 
  Since~$\elct$ realizes~$X$, for each $i \in [3m]$ there is exactly
  one vote in $\elct$ that ranks $\cand_{2i}$ below
  $\cand_{2i+1}$. This means that for each $u_i \in \calU$, there is
  exactly one set among $S_1, \ldots, S_m$ that includes $u_i$. Hence,
  $I$ is a \emph{yes}-instance of X3C.

  Finally, we observe that all votes in $\calD$ are single-peaked with
  respect to societal axis
  $c_{6m-1}\succ c_{6m-3} \succ \dots \succ c_{3}\succ
  c_{1}\succ c_{2}\succ c_4 \succ \dots \succ c_{6m}$ 
%
%
%
  and are balanced
  group-separable with respect to a balanced tree whose frontier is
  $c_1 \pref c_2 \pref \cdots \pref c_{6m}$ (to be formally correct,
  we would need to have a number of candidates equal to a power of two, it is
  easy to achieve this by adding at most $6m$ dummy candidates and extending 
  matrix $X$ so that for these candidates the diagonal entry would be equal to $m$).
\end{proof}

\paragraph{An Experiment.}
Using our algorithms from \Cref{thm:frecrec,thm:balanced}, we checked
for each of the frequency matrices from our 8x80~dataset whether it is
realizable by a single-peaked or a caterpillar/balanced
group-separable election (for each election we tried all societal axes
and all caterpillar trees).  For all three domains we found that for
each election in the dataset, its frequency matrix can be realized by
an election from the domain only if the election itself belongs to
this domain.\footnote{Elections from our dataset that are part of a
  restricted domain are almost exclusively sampled from models that
  are guaranteed to produce such elections.}
This
indicates that frequency matrices (and also position matrices) of elections from
restricted domains have specific features that are not likely to be produced by
elections sampled from other models.

\section{Condorcet Winners}
\label{sec:condorcet}
Our final set of results regards Condorcet winners in elections
that realize a given position matrix.

First, we consider the problem of deciding if 
a given position matrix
can be realized by an election
where a certain candidate is a Condorcet winner.
In general, the complexity of this problem remains open,
but if we restrict our attention to elections that only
contain votes from a given set
we obtain a hardness result
(even if 
the input matrix can always
be realized using votes from the given set).

\begin{restatable}{theorem}{condorcetnph}
\label{thm:condorcet-nph}
  Given a set $\calD$ of votes, listed explicitly,
  a position matrix $X$ (which can be realized by an election containing only votes from $\calD$)\footnote{Verifying this condition is not part of the problem as, by Theorem~\ref{thm:explicit}, such a test is $\np$-hard. It is simply a feature of our reduction.},
  and a candidate $c$,
  it is $\np$-hard to decide if there is an election realizing $X$,
  in which $c$ is a Condorcet winner and
  all votes come from $\calD$.
\end{restatable}

Making partial progress on the general problem,
we provide a necessary condition
for the existence of an election realizing a given position matrix
in which a given candidate $c$ 
is a Condorcet winner.
Roughly speaking, for each $i\in [m]$,
our condition looks for a set $S$ of candidates (different from $c$)
that frequently appear in the first $i$ positions.
If occurrences of candidates from $S$ on the first $i$ positions are ``sufficiently frequent''
compared to how often candidate~$c$ appears in the first $i-1$ positions,
and both $S$ and $i$ are ``small enough,''
then $c$ cannot be
a Condorcet winner in any election realizing the matrix.

\begin{restatable}{theorem}{condorcetcondition}
\label{thm:condorcet:condition}
For each position matrix $X$ and each $c \in [m]$,
if there is an election $\elct$ realizing $X$,
where $c$ is a Condorcet winner,
then for every $i \in [m]$ and $S \subseteq [m]$,
it holds that
\[
    \textstyle\sum_{j \in S} \!\sum_{k=1}^i \! X_{k,j} \!\le\!
    |S| \cdot\! \left\lfloor \!\frac{n-1}{2}\! \right\rfloor  + 
    \textstyle\sum_{k=1}^{i-1} \!\!\big( X_{k,c} \cdot \min(|S|,i-k) \!\big)\! .
\]
The condition can be checked in polynomial time.
\end{restatable}

\begin{figure}[t]
\centering
  \includegraphics[width=6cm]{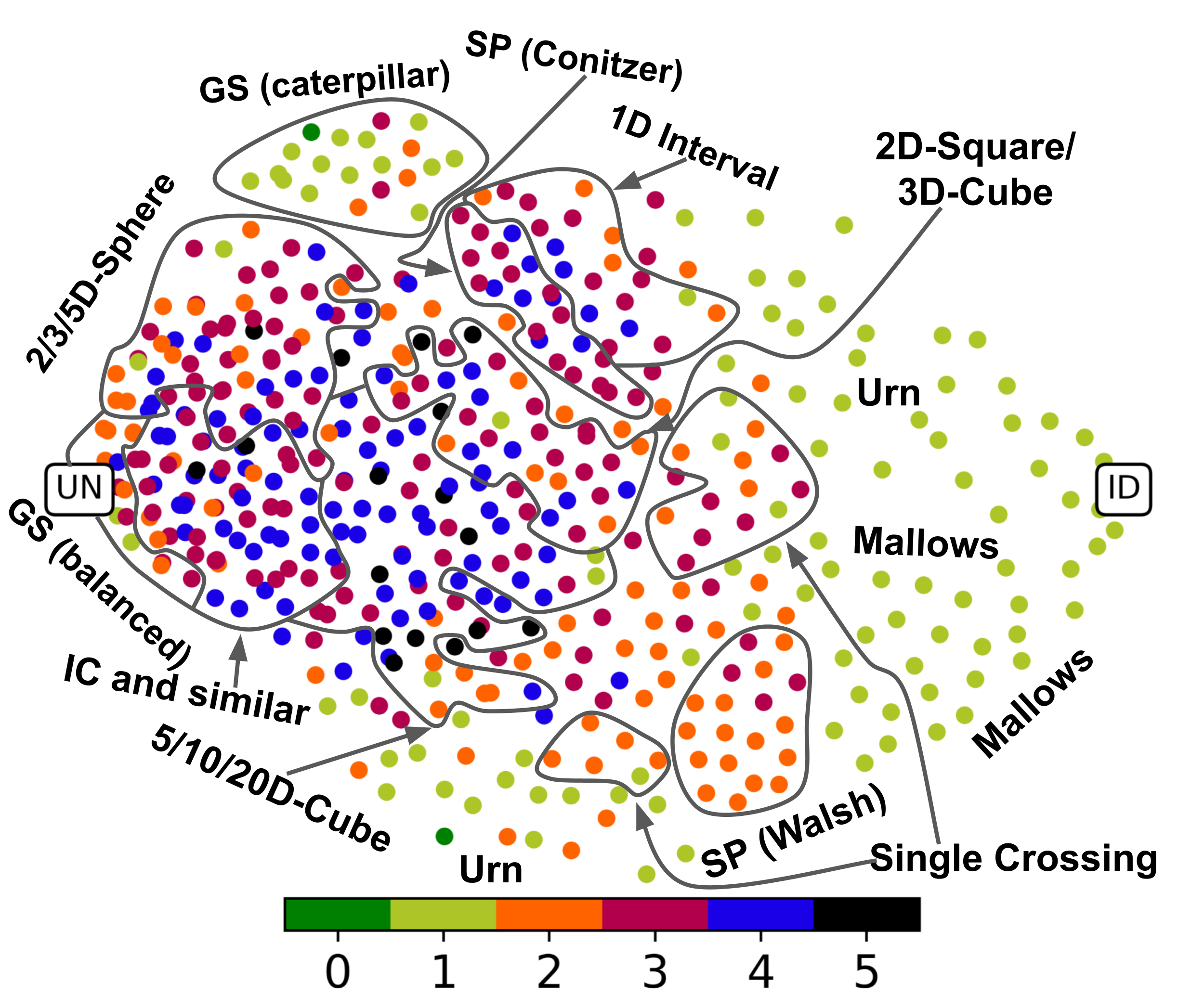}
  \caption{The number of candidates that can be Condorcet winners in elections realizing a given matrix in our 8x80 dataset.}
  \label{fig:condorcet_winners}
\end{figure}

\paragraph{Experiment 1.}
We tested our condition on the elections from the 8x80~dataset.
We checked for each election and each candidate whether the condition is satisfied,
but there is no election realizing the matrix in which
the candidate is a Condorcet winner (using an ILP formulation of the problem).
It turns out that this situation is very rare:
among all 480 matrices in the 8x80~dataset (i.e., among the position matrices of the elections
from the dataset),
there were only~6 in which there was one candidate
for which our condition gave the wrong answer
(there were none with more than one such candidate).
Thus, our condition appears to be quite an effective way
to detect potential Condorcet winners. 

\paragraph{Experiment 2.}
We conclude  with an experiment
where for each position matrix from our 8x80~dataset
we count how many different candidates are a Condorcet winner in at least one election realizing the matrix. 
The results are in \cref{fig:condorcet_winners}.

First, we observe that while $94$ elections from our 8x80~dataset do not have a Condorcet winner,  
only two of them have a position matrix that cannot be realized by an election with a Condorcet winner.
Second, examining  \cref{fig:condorcet_winners}, we see that for most matrices there are multiple different possible Condorcet winners, with the average number of Condorcet winners being $2.6$ and $120$ matrices having four or more possible Condorcet winners.
The number of possible Condorcet winners is correlated with the position of the matrix
on the map.  Generally speaking, it seems that the closer a
matrix is to UN, the more possible Condorcet winners we have. However,
in the close proximity of UN there is a slight drop in the
number of possible Condorcet winners.  Overall, these results confirm
that elections realizing a given position matrix can be very different
from each other (in terms of pairwise comparisons of candidates).


\section{Conclusions}
We have analyzed various properties of elections that realize given position or frequency matrices.
Among others, (i) we have shown algorithms for deciding if such elections can be implemented
using votes from particular structured domains, and 
(ii) we have found that for a given matrix, such elections can be very diverse.
The latter result is witnessed by the fact that two elections realizing a matrix
may have large isomorphic swap distance and  may have different
Condorcet winners. Hence, while maps of elections (based on position matrices)
certainly are very convenient tools for visualizing some experimental results (including ours), for others their value might be limited. It would be interesting to find  such experiments and establish their common features.

\section*{Acknowledgments}
NB was supported by the DFG project ComSoc-MPMS (NI 369/22).
This project has received funding from the European Research Council (ERC) under the European Union’s Horizon 2020 research and innovation programme (grant agreement No 101002854).
\begin{center}
  \includegraphics[width=3cm]{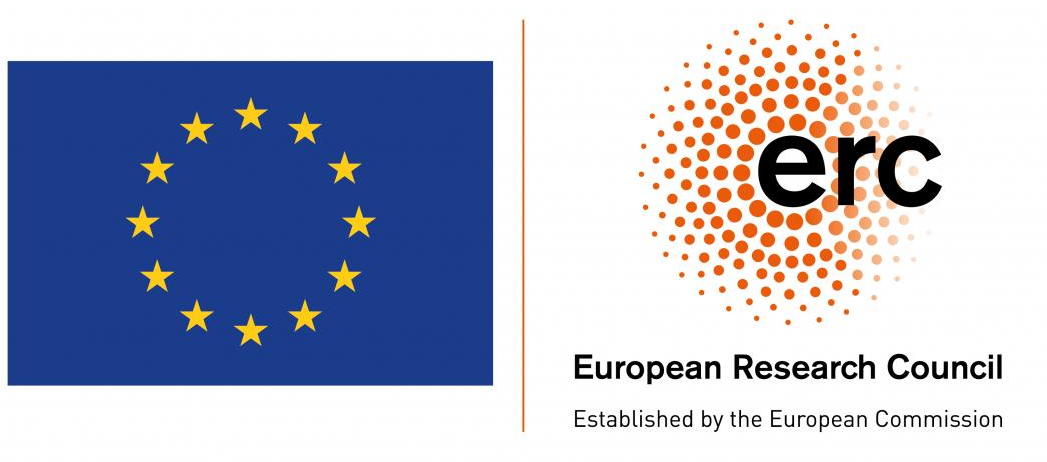}
\end{center}

\bibliographystyle{plainnat}
\bibliography{bib}

\appendix

\section*{Appendix}
\section{Map of Elections}\label{app:map}

Given two vectors $x, y \in
\mathbb{R}^n$ 
of nonnegative numbers that both sum up to the same value, their
Earthmover's distance, denoted $\emd(x,y)$, is the smallest total cost
of transforming one to the other by using the following operations:
For each $i, j \in [n]$ and $\delta > 0$, if the $i$-th entry of $x$
is at least $\delta$, then at cost $\delta \cdot |i-j|$ we can
subtract $\delta$ from it and add it to its $j$-th entry.

Given a matrix $X$, by $X_{j}$ we mean the row vector
equal to $X$'s $j$-th column.
Let $\elct$ and $\flct$ be two elections, each with~$m$ candidates
and~$n$ voters. Let $E$ and $F$ be their (arbitrarily chosen) position
matrices. The positionwise distance of $\elct$ and $\flct$ is:
\[
\textstyle  d_\pos(\elct,\flct) = \min_{\sigma \in S_m} \sum_{j=1}^m \emd( E_j, F_{\sigma(j)})
\]
(note
that due to minimization over permutations $\sigma$, the distance does
not depend on the choice of the position matrices).

\subsection{Single-Crossing Elections}
Among many types of elections introduced in the main part, our map
additionally contains \emph{single-crossing elections}. Similarly to
single-peaked ones, they can also be understood as describing the
left-to-right political spectrum, but this time from the perspective
of the voters.
\begin{definition}\label{def:singlecrossing}
  An election~$\elct{} = (\cnds, \vtrs)$ with $\vtrscnt$~voters
  ordered~$(\voter_1, \voter_2, \ldots, \voter_\vtrscnt)$ is
  \emph{single-crossing with respect to this order} if for each
  pair~$(\cand, \cand')$ of the candidates such that
  $\cand \succ_{v_1} \cand'$, there exists a positive integer~$t$ such
  that
  $\{ i \in \{1, 2, \ldots n\} \mid \cand_1 \succ_{v_i} \cand_2\} = \{1,
  2, \ldots, t\}$.  An election~$\elct{} = (\cnds, \vtrs)$ is
  \emph{single-crossing} if there exists an ordering of the voters
  with respect to which the election is single-crossing.
\end{definition}
\noindent In single-crossing elections we can order the voters in a
way that for each pair of candidates their relative order changes at
most once when sweeping through the order votes.

\subsection{Statistical Models}
To generate elections, we use various statistical models, which we describe
below.
\begin{description}
 \item[Impartial Culture] We draw each voter's preference order uniformly
 at random from the collection of all possible preference orders.

 \item[P\'olya-Eggenberger Urn Model] In the Urn model~\citep{berg1985paradox},
 we start with an urn containing every possible preference order. Constructing
 an election, we draw its votes iteratively (starting with an empty collection
 of votes), one vote per iteration. For some fixed parameter $\alpha \in [0,
 1]$, in each iteration, we first draw a new voter's preference order uniformly
 at random from the urn. Then, we return it to the urn together with  $\alpha
 \cndscnt!$~copies of the drawn preference order. For~$\alpha = 0$, the Urn model is
 equivalent to the Impartial Culture.

 \item[Mallows Model] The Mallows model~\citep{mal:j:mallows} is parameterized by a central preference
 order~$u^*$ and a dispersion parameter~$\phi \in [0, 1]$. For fixed values of
 the parameters, we draw each preference order of a generated election
 independently. The probability of drawing a given preference order~$u$ is
 proportional to $\phi ^{\swap(u, u^*)}$ (recall that~$\swap(u,
 u^*)$ is the number of swaps required to transform $u$ to~$u^*$). For the
 special cases of~$\phi = 0$ and~$\phi = 1$, we obtain, respectively, an
 election with all voters having the central preference order and the IC model.

 \item[Normalized Mallows Model] This model is an adaptation of the Mallows
 model~\citep{boe-bre-fal-nie-szu:c:compass}. Instead of the dispersion
 parameter~$\phi$, it uses as a parameter the expected relative swap
 distance~$\textrm{rel-}\phi \in [0,1]$ of~$u^*$ from a drawn vote. Given some value
 of~$\textrm{rel-}\phi$, we first compute the value of the dispersion
 parameter~$\phi$ of the (standard) Mallows Model that yields the requested
 expected relative swap distance~$\textrm{rel-}\phi$. Then, we use this
 parameterization of the (standard) Mallows Model, with the given central
 preference, to generate an election.

\item[Euclidean Models] In Euclidean models~\citep{enelow1984spatial,
    enelow1990advances}, we model the ideological space as the
  Euclidean space.  Specifically, for a $t$-dimensional Euclidean
  model, we generate voters and candidates as points in the
  $t$-dimensional Euclidean space. Then, we construct the preference
  order of a voter by listing the candidates from the closest one to
  the farthest one. We consider two ways of generating the points:
  \begin{itemize}
   \item \textbf{Uniform Interval Model:} We place each point uniformly at random in a $t$-dimensional
   hypercube $[0,1]^t$.
   \item \textbf{Sphere Model:} We place each point uniformly at random on a $t$-dimensional
   hypersphere with radius~$1$ centered at point~$(0, 0, \ldots, 0$).
  \end{itemize}

 \item[Single-Peaked Elections Models] We consider two models, one
 by~\citet{con:j:eliciting-singlepeaked} and one
 by~\citet{wal:t:generate-sp}; hence, we call them the Walsh and the Conitzer models. Under both of
 them, we first randomly select the societal axis; for the
 sake of presentation, assume it to be~$c_1 \mathrel\triangleright c_2
 \mathrel\triangleright \ldots \mathrel\triangleright c_\cndscnt$. Then we
 proceed as follows:
  \begin{itemize}
  \item \textbf{Conitzer Model:} To draw a preference order, we first
    select the top candidate~$c_i$ of the order. Then, we fill up the
    preference order in $\cndscnt - 1$~steps as follows. In each step,
    assuming that the preference order already consists of candidates
    from $c_j$ to $c_{j'}$, $j < j'$, with the same probability we
    extend the preference order by either $c_{j-1}$ or~$c_{j' +
      1}$. It might happen that one of these candidates does not exist
    and then we take the existing one.
   \item \textbf{Walsh Model:} We draw preference orders uniformly at
   random from the space of all possible preference orders (conforming the
   selected axis), as described by~\citet{wal:t:generate-sp}.
  \end{itemize}

\item[Single-Peaked on a Circle Elections Model] Following
  \citet{szu-fal-sko-sli-tal:c:map}, to generate single-peaked on a
  circle vote we use the same procedure as for the Conitzer model,
  except that we take care of the fact that the societal axis is
  cyclical.

 \item[Group-Separable Elections Models] We only generate group-separable elections
 based on caterpillar and balanced trees. Doing so, we take the approach
 of~\citet{boe-fal-nie-szu-was:c:metrics}. Consider an initial (respective) tree
 in which each leaf represents a unique candidate such that the leftmost leaf
 represents~$c_1$, the next leaf to the right represents~$c_2$, and so on. To
 generate a vote, we start from the initial tree and reverse the order of each
 node's children with probability $\frac{1}{2}$. Then, the order of the
 candidates, from left to right, represents the preference order of the new
 vote.

\item[Single-Crossing Elections Model] We use this distribution only
  for the sake of completeness and compatibility with other datasets
  from the literature.  Thus, we will skip the description of this
  model referring the reader to the work
  of~\citet{szu-fal-sko-sli-tal:c:map}, whose procedure we use.
\end{description}

\subsection[Description of the 8x80 Dataset]{Description of the 8x80~Dataset}
The dataset contains~$480$~elections generated according to the models
introduced in the previous section. The exact combination of election models and
the quantities of elections generated using them is depicted
in~\cref{tbl:elections-blend}.

\begin{table}
\centering
\caption{The ingredients of the 8x80~dataset grouped by the election
models. \label{tbl:elections-blend}}
\begin{tabular}{rl}
 \toprule
 model & \#elections\\\midrule
 Impartial Culture & 20\\
 single-peaked (Conitzer) & 20\\
 single-peaked (Walsh) & 20\\
 single-peaked on a circle & 20\\
 single-crossing & 20\\
 1D-Euclidean (uniform interval) & 20\\
 2D-Euclidean (uniform interval) & 20\\
 3D-Euclidean (uniform interval) & 20\\
 5D-Euclidean (uniform interval) & 20\\
 10D-Euclidean (uniform interval) & 20\\
 20D-Euclidean (uniform interval) & 20\\
 2D-Euclidean (sphere) & 20\\
 3D-Euclidean (sphere) & 20\\
 5D-Euclidean (sphere) & 20\\
 group-separable (balanced) & 20\\
 group-separable (caterpillar) & 20\\
 normalized Mallows model & 80\\
 urn model & 80\\\bottomrule
\end{tabular}
\end{table}

For the $80$~elections generated using the urn model and the normalized Mallows
model, we followed the protocol
of~\citet{boe-bre-fal-nie-szu:c:compass,boe-fal-nie-szu-was:c:metrics}. Hence,
for each of the elections generated with the normalized Mallows Model, we drew
the value of $\textrm{rel-}\phi$ uniformly at random from the $[0,1]$~interval.
The parameter for the urn model elections was drawn according to the Gamma
distribution with the shape parameter $k = 0.8$ and the scale parameter
$\theta = 1$. We refer to the work of~\citet{boe-bre-fal-nie-szu:c:compass} for
a detailed discussion on the presented choice of parameters.

\section[Additional Material for Section 3]{Additional Material for \Cref{sec:counting}}
\subsection{Uniform Sampler}\label{app:sampler}

In this section we describe an algorithm that given a position matrix
$X$ samples an election that realizes it uniformly at random. The
algorithm relies on the ability to count elections that realize a
given position matrix and whose lexicographically first vote has a
given prefix (if the prefix is empty, then this problem becomes
\realizationsProb{}).  We refer to this counting problem as
\lexrealizationsProb{}.

Let $X$ be our input $m \times m$ position matrix, whose each row and
each column sums up to $n$. We let the candidate set be
$\cnds = \{\cand_1, \ldots, \cand_m\}$. We assume the natural ordering
over the candidates, so $\cand_1$ corresponds to the first column of
$X$, $\cand_2$ corresponds to the second one, and so on. We also use
this order to compare votes (so, e.g., vote
$c_3 \pref c_1 \pref c_2 \pref \cdots$ precedes vote
$c_3 \pref c_1 \pref c_4 \pref \cdots$).  Our sampler generates the
output election vote by vote, in a lexicographic order. Initially, we
assume that there is only one ``virtual'' vote
$v_0 \colon c_1 \pref c_2 \pref \cdots \pref c_m$. Our algorithm will
generate votes $v_1, \ldots, v_n$ such that
$v_0 \leq v_1 \leq v_2 \leq \cdots \leq v_m$ (for two votes $x$ and
$y$, by $x \leq y$ we mean that $y$ is lexicographically greater or
equal to $x$).

Below we describe the process of generating the votes. Let us say that
we have already generated votes $v_0, \ldots, v_{i-1}$ and currently
the goal is to generate $v_i$. Let $Y_{i-1}$ be the position matrix
(for the natural ordering of the candidates) of election
$(C,(v_1,\ldots, v_{i-1}))$ and let $X_{i-1} = X - Y_{i-1}$.
Let~$\elct_{i}$ be a random variable equal to an election that
realizes~$X_{i-1}$ and is selected uniformly at random among such
elections whose lexicographically first vote is lexicographically
greater or equal to $v_{i-1}$. We write $u$ to denote the (random
variable equal to the) lexicographically first vote in
$\elct_{i-1}$. For each $\ell \in [m-2] \cup \{0\}$ we define the
random event $T_\ell$ such that:
\begin{enumerate}
\item length-$\ell$ prefix of $u$ is equal to the length-$\ell$ prefix
  of $v_{i-1}$,
\item lenght-$(\ell+1)$ prefixes of $u$ and $v_{i-1}$ are different,
  and
\item $u$ is lexicographically greater than $v_{i-1}$.
\end{enumerate}
Additionally, we let $T_{m-1}$ be the event that $u = v_1$. Note that
events $T_0, \ldots, T_{m-1}$ partition the space of elections from
which $\elct_i$ is chosen. Further, computing the probability of each
of these events is easy, provided that we have an algorithm for
solving \lexrealizationsProb{} (although for each $T_\ell$ we might
have to invoke this algorithm up to $O(m)$ times).

The first step of generating $v_i$ is to choose a value
$\ell \in [m-1] \cup \{0\}$ with probability equal to that of event
$T_\ell$. If $\ell = m-1$ then we let $v_i = v_{i-1}$ and we proceed
to generating $v_{i+1}$. Otherwise, we fix the length-$\ell$ prefix of
$v_i$ to be equal to the length-$\ell$ prefix of $v_{i-1}$ and we
generate the reminder of the vote as follows (in a
candidate-by-candidate manner). Let $c$ be the canidate that $v_{i-1}$
ranks on position $\ell+1$ and let $D$ be the set of candidates $d$
such that (a)~$v_{i-1}$ ranks $d$ on a position greater than $\ell+1$
and (b)~$d$ is greater than $c$ in our candidate ordering (at least
one such candidate exists because otherwise the probability of
selecting this value of $\ell$ would be $0$). Then, we select a
candidate $d \in D$ with probability equal to that of choosing
uniformly at random an election $\elct'_i$ that realizes $X_i$ and
whose lexicographically first vote has $\ell+1$-length prefix equal to
length-$\ell$ prefeix of $v_{i-1}$ extended with $d$ (again, we can
compute these probabilities using \lexrealizationsProb{}).  We extend
$v_{i}$ with $d$. We continue choosing the candidates for the
following positions of $v_i$ in the same way, except that now in each
iteratioin we let $D$ be the set of candidates not-yet-ranked within
$v_i$ (for the $(\ell+1)$-st position we had to be more careful to
ensure that the length-$(\ell+1)$ prefix of $v_i$ is lexicographically
greater than the length-$(\ell+1)$ prefix of $v_{i-1}$).  This
completes the description of the sampler.

The algorithm generates elections that realize $X$ uniformly at random
because each random decision corresponds to partitioning the space of
elections that generate $X$, and we make each decision with
probability proportional to the size of the subspace to which we
restrict our attention.

\subsection{Experiments on 4x16 Dataset}\label{sub:416exmperiments}

\begin{figure*}
 \def\wdtratio{.22}
 \begin{subfigure}[t]{\wdtratio\textwidth}
 \includegraphics[width=4.5cm]{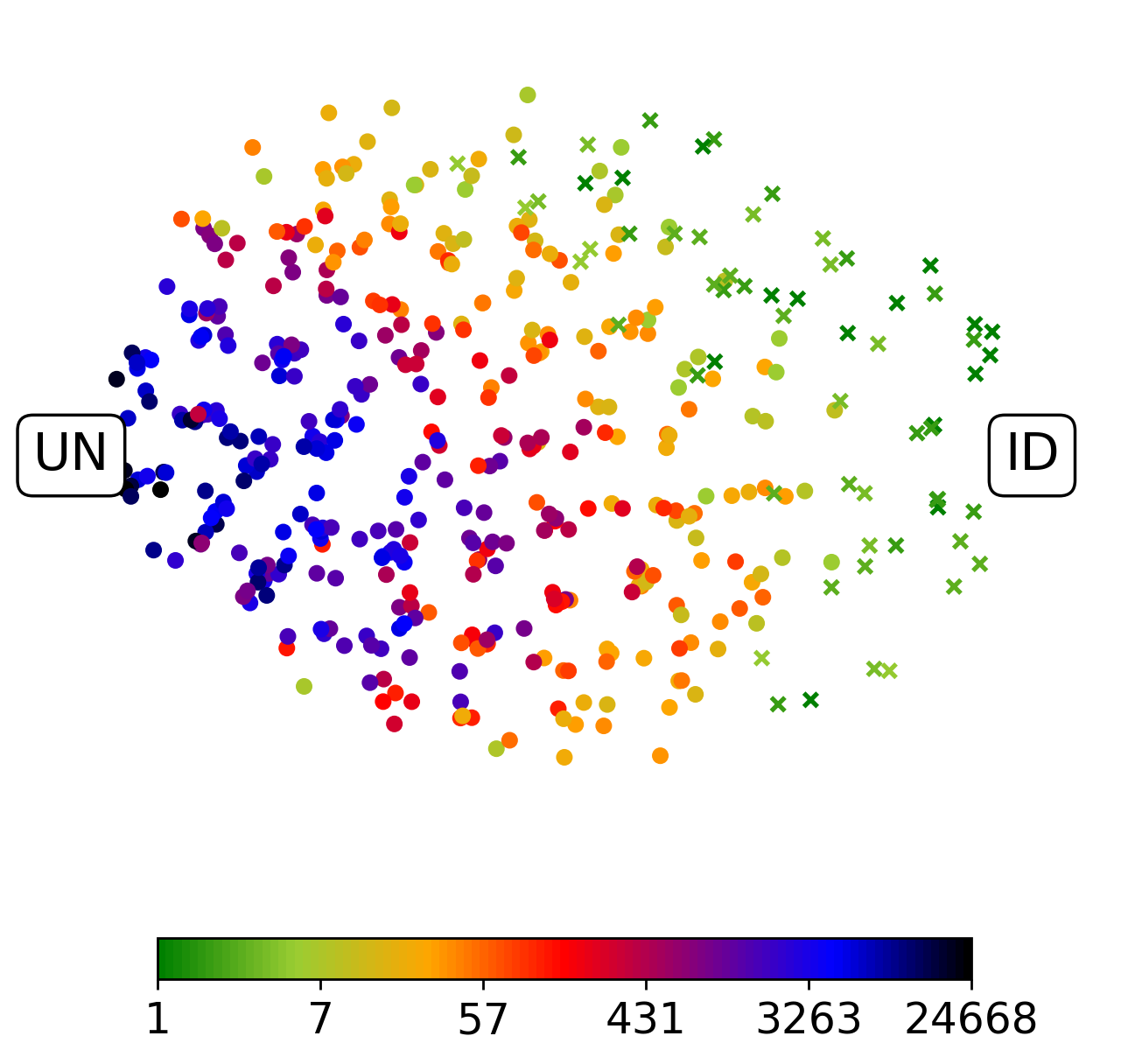}
 \caption{Number of realizations.}\label{fig:real_counts-4x16-map}
 \end{subfigure}\hfill
 \begin{subfigure}[t]{\wdtratio\textwidth}
 \includegraphics[width=4.5cm]{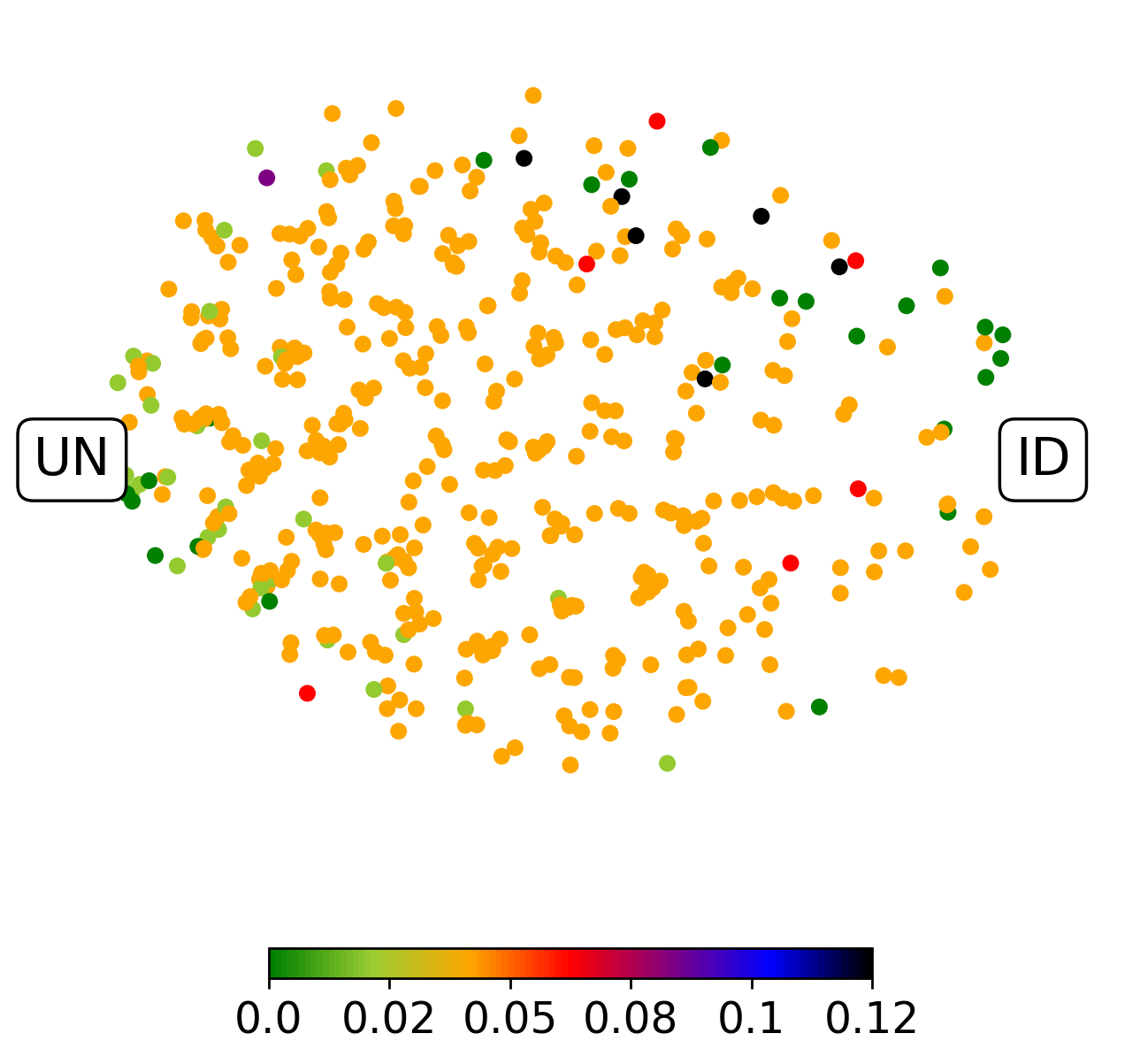}
 \caption{Minimum swap distance.}\label{fig:minswap-4x16-map}
 \end{subfigure}\hfill
 \begin{subfigure}[t]{\wdtratio\textwidth}
 \includegraphics[width=4.5cm]{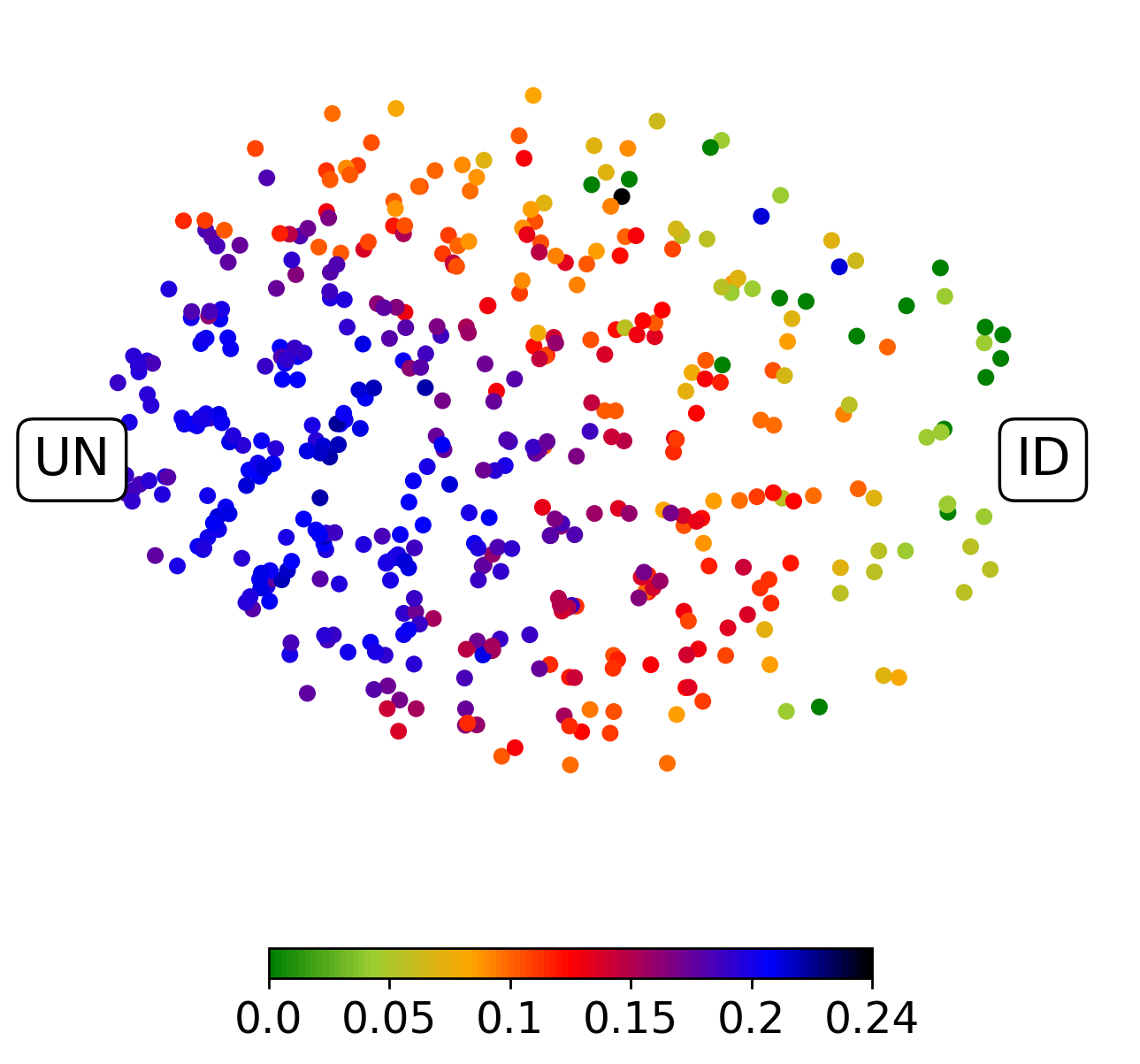}
 \caption{Average swap distance.}\label{fig:avgswap-4x16-map}
 \end{subfigure}\hfill
 \begin{subfigure}[t]{\wdtratio\textwidth}
 \includegraphics[width=4.5cm]{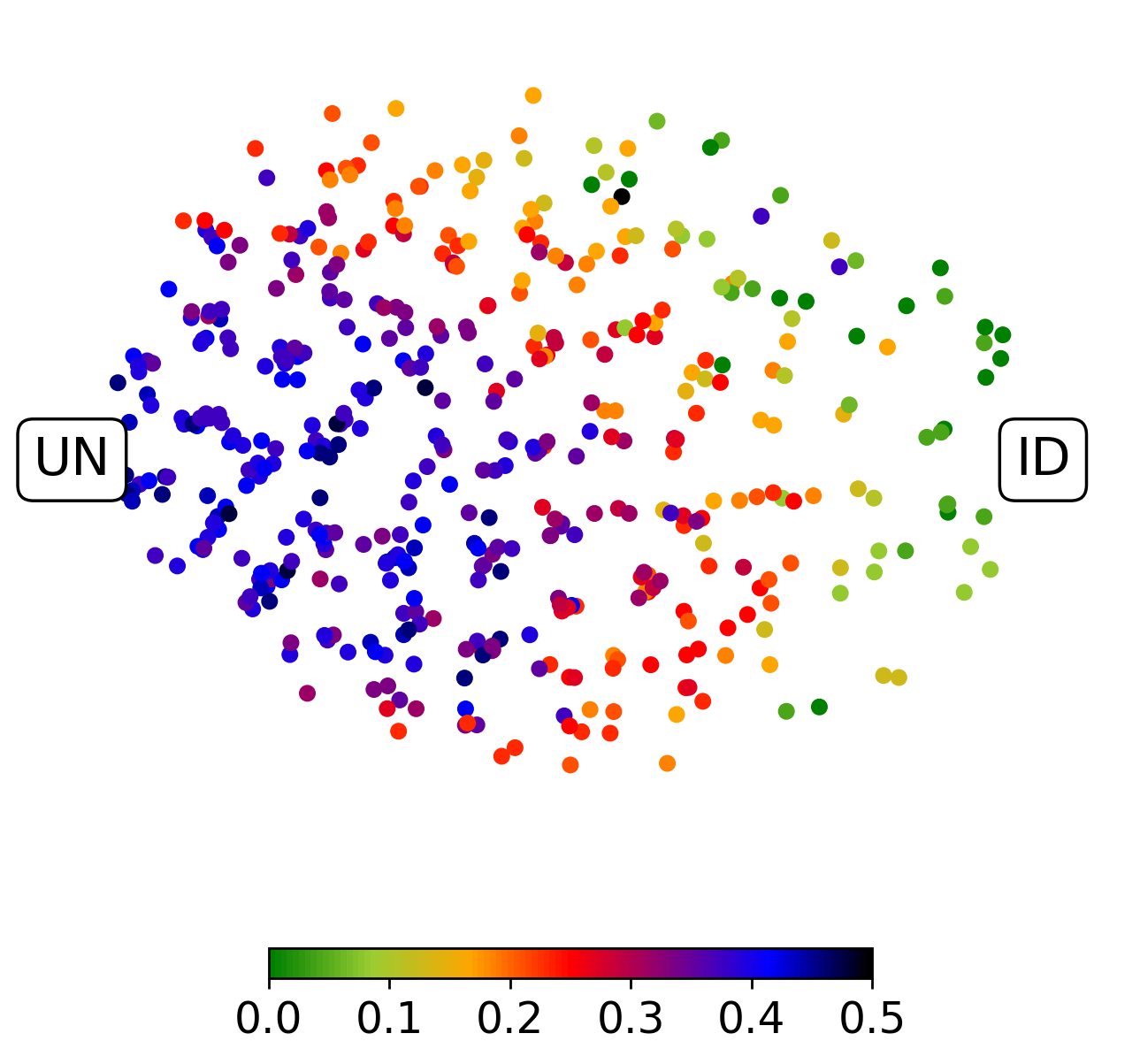}
 \caption{Maximum swap distance.}\label{fig:maxswap-4x16-map}
 \end{subfigure}
 \caption{Maps with our experimental results for the 4x16
    dataset. In~\cref{fig:real_counts-4x16-map}, the dot colors describe the
    number of realizations of the position matrix of the respective election.
    The crosses indicate that there were at most 5~such realizations.
    In~\cref{fig:minswap-4x16-map,fig:avgswap-4x16-map,fig:maxswap-4x16-map},
    the color shows, respectively, the minimum, average, and maximum isomorphic
    swap distance between all pairs of the elections realizing the position
    matrix of the respective election.\label{fig:4x16-distance-results}}
\end{figure*}

We generated a 4x16~dataset used in this experiment analogously to the~8x80~one
introduced in the main part of the paper~(\cref{sec:map-of-elections}). That is,
the 4x16~dataset includes $480$~elections with $4$~candidates and~$16$~voters;
the election distributions and their counts are shown
in~\cref{tbl:elections-blend}.

We implemented a naive ILP program to compute all elections realizing a given
position matrix. In the ILP program, we specified an integer variable for each
possible $24$~preference orders over four candidates. The value of each of these
variables in a feasible solution corresponded to the number of voters with a
given preference orders. Using constrains, we ensured that there are exactly
$16$~voters and that their preference orders conform to the given position
matrix. Naturally, this approach does not scale well with the election size.
Indeed, getting all elections realizing matrices with six voters turned out to
be too computationally intensive in practice.

We performed two experiments on the 4x16 dataset. In the first experiment, for
each election of the dataset we computed its position matrix and all elections
realizing this matrix. We present the number of such realizing elections on the
map in~\cref{fig:real_counts-4x16-map}. In the second experiment, for each
position matrix in our dataset, we took all elections realizing the matrix and
we computed isomorphic swap distances between all pairs of them. On the maps
in~\cref{fig:minswap-4x16-map,fig:avgswap-4x16-map,fig:maxswap-4x16-map}, we
present the minimum, the maximum, and the average over these distances (for each
matrix separately).

The results regarding the average and maximum
distances~(\cref{fig:avgswap-4x16-map,fig:maxswap-4x16-map}) show that
the matrices of elections closer to ID tend to yield elections with a
very small average and maximum pairwise isomorphic distances. However,
the distances increase as we get closer to UN. Indeed, elections near
UN achieve the greatest value of the maximum distance of around
$50\%$~of the maximum achievable distance between two elections of
four candidates and 16~voters. For each election in the dataset, a
high value of the average and maximum distances coincide with a high
number of elections realizing the election's position
matrix~(\cref{fig:real_counts-4x16-map}) and vice versa. This is
intuitive, as one would expect that with a greater diversity of
elections realizing a single matrix, the average distance among them
rises. In general, the results for the 4x16 dataset are qualitatively
very similar to the results presented in~\cref{sec:distance-exp} for
the 8x80~dataset.

For the minimum distance we observe a (roughly) same-colored
map, with almost all values indicating the minimum distance below~$5\%$ of the maximum
achievable distance. This shows that irrespectively of the considered matrix, there
always were two elections realizing a given matrix that were relatively close to
each other with respect to the isomorphic swap distance.

\section[Additional Material for Section 4]{Additional Material for \cref{sec:structure}}
In this section,
we provide the proofs of \cref{thm:frecrec,,thm:balanced},
which were omitted in the main body of the paper.

\subsection[Proof of Theorem 3]{Proof of \cref{thm:frecrec}}\label{sec:proof_freq_matrices}

\frecrec*
\begin{proof}
We split the proof,
into four independent propositions,
which we will subsequently prove:
\begin{itemize}
    \item in \cref{prop:input-given-p} we focus on explicit domains $\calD$,
    \item in \cref{prop:balanced:fixed:p} we turn to elections in
    group-separable domain $\calD$ with balanced trees,
    \item in \cref{prop:caterpillar:fixed:freq:p} we consider
    group-separable domain $\calD$ with caterpillar trees, and
    \item in \cref{prop:single-peaked:fixed:freq:p} we look at
    single-peaked domain $\calD$.
\end{itemize}

We will proceed with the propositions one-by-one.
Let us start with domains explicitly given as a set of votes.

\begin{proposition}\label{prop:input-given-p}
    There is a polynomial-time algorithm
    that given a frequency matrix $X$
    and an explicit domain $\calD$,
    decides whether there is an election that realizes $X$,
    and whose all votes belong to $\calD$.
\end{proposition}
\begin{proof}
  In short, if $D = \{v_1, \ldots, v_n\}$ then it suffices to find rational
  values $y_1, \ldots, y_n$ such that:
  \[
    X = y_1 P(v_1) + y_2 P(v_2) + \cdots + y_n P(v_n).
  \]
  Doing so in polynomial time is a simple linear programming task.
\end{proof}

Now, let us focus on group-separable domains.
First let us consider these with a balanced tree.
Here, our proof works in fact for both
frequency and position matrices.

\begin{proposition}\label{prop:balanced:fixed:p}
    There is a polynomial-time algorithm
    that given a frequency or position matrix $X$
    and a balanced group-separable domain $\calD$,
    decides whether there is an election that realizes $X$,
    and whose all votes belong to $\calD$.
\end{proposition}
\begin{proof}
First, let us consider position matrices and
we will move to frequency matrices at the end of the proof.

For an arbitrary $2m \times 2m$ matrix $A = [a_{i,j}]_{i=1,j=1}^{2m,2m}$,
by a \emph{quarter} of $A$,
let us denote each of the four $m \times m$ matrices
obtained by ``cutting'' $A$ in the middle horizontally and vertically,
i.e., matrices
$[a_{i,j}]_{i=1,j=1}^{m,m}$,
$[a_{i,j}]_{i=m+1,j=1}^{2m,m}$,
$[a_{i,j}]_{i=1,j=m+1}^{m,2m}$, and
$[a_{i,j}]_{i=m+1,j=m+1}^{2m,2m}$.
Now, we will say that $A$ is \emph{evenly-quartered}
if the sum of entries in its \emph{upper-left} quarter
is equal to the sum of entries in its \emph{bottom-right} quarter and
the sum of entries in its \emph{upper-right} quarter
is equal to the sum of entries in its \emph{bottom-left} quarter.
Formally,
\begin{multline*}
    \sum_{i=1}^m \sum_{j=1}^m a_{i,j} =
    \sum_{i=m+1}^{2m} \sum_{j=m+1}^{2m} a_{i,j}
    \quad \mbox{and} \\
    \sum_{i=m+1}^{2m} \sum_{j=1}^m a_{i,j} =
    \sum_{i=1}^m \sum_{j=m+1}^{2m} a_{i,j}.
\end{multline*}
Finally,
we will say that $m \times m$ matrix $A$ is
\emph{maximally-evenly-quartered} if
$m=2^k$ for some $k \in \mathbb{N}$ and
$A$ is either a one-element-matrix or
it is evenly-quartered and
each of its quarters is maximally-evenly-quartered.

Let $X$ be an arbitrary $m \times m$ matrix.
Without loss of generality,
we will focus on balanced group-separable elections,
$\elct = (C,V)$, realizing $X$ and compatible with
the binary tree, $\calT$, in which the order of the candidates at leaves
aligns with the order of the corresponding columns in $X$
(otherwise we can we can reorder the columns of $X$ accordingly).
We will prove that such election exists, if and only if,
$X$ is maximally-evenly-quartered
(note that in this way we will prove also that
each maximally-evenly-quartered matrix is a position matrix).
Since checking if matrix is evenly-quartered can be done in polynomial time
and there are $O(m^2)$ matrices
that can be obtained from $X$ by the sequence of taking the quarters
($(4m^2 - 1)/3$ to be exact),
this will imply the thesis.

First let us show that if $X$ is a position matrix
of some balanced group-separable election $\elct = (C,V)$,
then $X$ is maximally-evenly-quartered.
Since $\elct$ is balanced group-separable,
$|C|=2^k$ for some $k \in \mathbb{N}$.
Let us proceed by an induction on $k$.
If $k=0$, the thesis holds trivially.
Assume $k>0$.
Let us denote $C=\{c_1,\dots,c_{2^k}\}$.
In order to prove that $X$ is maximally-evenly-quartered,
we have to prove that
(1) $X$ is evenly-quartered and
(2) each of its quarters is maximally-evenly-quartered.

(1) Let us start by proving that $X$ is evenly-quartered.
To this end,
let us denote $C_1 = \{c_1,\dots,c_{2^{k-1}}\}$ and
$C_2 = \{c_{2^{k-1}+1},\dots,c_{2^k}\}$.
Since $\elct$ is balanced group-separable,
we can split the voters in two disjoint sets,
$V_1, V_2 \subseteq V$,
such that $V_1 \cup V_2 = V$ and
voters in $V_1$ prefer each candidate in $C_1$
over each candidate in $C_2$ and,
conversely,
voters in $V_2$ prefer each candidate in $C_2$
over each candidate in $C_1$.
Observe that in both upper-left and bottom-right quarters of $X$,
the sum of entries is equal to $|V_1|$.
Analogously, in both upper-right and bottom-left quarters of $X$,
the sum of entries is equal to $|V_2|$.
Hence, $X$ is evenly-quartered.

(2) Now, let $Y$ be an arbitrary quarter of $X$.
We will construct elections $E'=(V',C')$ for which $Y$ is a position matrix.
First, let us take $C'=C_1$,
if $Y$ is upper- or bottom- left quarter,
and $C'=C_2$,
otherwise.
Similarly,
for $V'$ let us take the set of all votes from $V_1$ restricted to $C'$,
if $Y$ is upper-left or upper-right quarter of $X$,
and all votes from $V_2$ restricted to $C'$,
otherwise.
Then, $E'=(V',C')$ is also a balanced group-separable election,
compatible with one of the two main branches of the original tree of $E$.
Since $|C'|=2^{k-1}$, by the inductive assumption,
$Y$, is maximally-evenly-quartered.
Therefore, $X$ indeed is maximally-evenly-quartered.

\newcommand{\ul}{\textrm{ul}}
\newcommand{\ur}{\textrm{ul}}
\newcommand{\bl}{\textrm{ul}}
\newcommand{\br}{\textrm{ul}}
In the remainder of the proof,
let us show that if position matrix $X$ is maximally-evenly-quartered,
then there exists a balanced group-separable election realizing $X$.

We start by introducing some additional notation. For a vote~$v$ over
candidates~$D$ and a vote~$u$ over candidates~$D'$ such that $D \cap D' =
\emptyset$, by $v \circ u$ we denote the \emph{concatenation} of $v$ and $u$.
That is, the vote in which each candidate from~$D$ is preferred over each
candidate from~$D'$ and each pair of candidates from one of the sets~$D$ or~$D'$
is ordered in the same way as in vote $v$ or $u$, respectively.

We now present how to build a balanced group-separable election realizaing~$X$.
Denote
the upper-left, upper-right, bottom-left, bottom-right quarters of $X$
by $Y^{\ul}$, $Y^{\ur}$, $Y^{\bl}$, and $Y^{\br}$, respectively.
Since $X$ is maximally-evenly-quartered,
each of its quarter is also maximally-evenly-quartered.
Hence, from the inductive assumption,
for each quarter,
there exists a balanced group-separable election
with a binary tree in which candidates at consecutive leaves
correspond to consecutive columns in this quarter.
Let us denote such elections by
$E^{\ul} = (V^{\ul},C^{\ul})$,
$E^{\ur} = (V^{\ur},C^{\ur})$,
$E^{\bl} = (V^{\bl},C^{\bl})$,
and $E^{\br} = (V^{\br},C^{\br})$
for matrices $Y^{\ul}$, $Y^{\ur}$, $Y^{\bl}$, and $Y^{\br}$, respectively.
Since quarters $Y^{\ul}$ and $Y^{\bl}$
share the same columns in the original matrix $X$,
let us denote $C = C^{\ul} = C^{\bl}$ and, analogously,
$C' = C^{\ur} = C^{\br}$.

Now, observe that since $X$ is evenly-quartered,
we get that $|V^{\ul}| = |V^{\br}|$.
Let us then denote
$V^{\ul} = \{v^{\ul}_1,\dots,v^{\ul}_{n_1}\}$ and
$V^{\br} = \{v^{\br}_1,\dots,u^{\br}_{n_1}\}$.
Then, we construct the following set of votes:
$$V = \{ v^{\ul}_i \circ v^{\br}_i : i \in [n_1]\}.$$
Analogously, $|V^{\ur}| = |V^{\bl}|$.
Hence, let us denote
$V^{\ur} = \{v^{\ur}_1,\dots,v^{\ur}_{n_2}\}$ and
$V^{\bl} = \{v^{\bl}_1,\dots,u^{\bl}_{n_2}\}$.
Again, we construct the following set of votes:
$$V' = \{ v^{\ur}_i \circ v^{\bl}_i : i \in [n_2]\}.$$
Eventually, we construct election~$E = (V \cup V', C \cup C')$.
Observe that because of our construction,
the position matrix of~$E$ is equal to $X$.
Moreover,
every voter in $V$ prefers each candidate in $C$
over each candidate in $C'$.
Similarly,
every voter in $V'$ prefers each candidate in $C'$
over each candidate in $C$.
Hence,
since $E^{\ul}$, $E^{\ur}$, $E^{\bl}$, and $E^{\br}$
are all balanced group-separable elections,
$E$ is also balanced group-separable and the thesis follows.

Finally,
let us consider the frequency matrix case.
Observe that if in a maximally-evenly-quartered matrix
we multiply each element by a constant,
the resulting matrix is still maximally-evenly-quartered.
Hence, we can generalize our maximally-evenly-quartered property
also to frequency matrices,
by requiring the same equalities to hold.
Since for every frequency matrix $X$ there exists a constant $n$
such that matrix~$X$ with each element multiplied by $n$ is a position matrix,
our proof holds also for frequency matrices.
\let\ul\undefined
\let\ur\undefined
\let\bl\undefined
\let\br\undefined
\end{proof}

Now, let us move to caterpillar group-separable domains.

\begin{proposition}\label{prop:caterpillar:fixed:freq:p}
    There is a polynomial-time algorithm
    that given a frequency matrix $X$
    and a caterpillar group-separable domain $\calD$,
    decides whether there is an election that realizes $X$,
    and whose all votes belong to $\calD$.
\end{proposition}
\begin{proof}
We will show that deciding
whether for a given frequency matrix $X$
there exists a caterpillar group-separable election profile $\elct$
with a given tree realizing $X$
is equivalent to deciding whether a certain linear programming problem
has a solution.
Since the latter is known to be polynomial,
we will obtain the thesis.
Without loss of generality,
we focus on caterpillar elections
compatible with a caterpillar tree $\calT$
in which candidates at consecutive leaves
correspond to consecutive columns of matrix $X$
(otherwise the columns of~$X$ can be rearranged).

For every $k,\ell \in \mathbb{Z}$,
by $[k,\ell]$ let us denote set $\{k,k+1,\dots,\ell\}$.
Fix an arbitrary frequency matrix $X$. To define our linear program,
for each column $j \in [m]$ and row $i \in [m]$, we introduce
two rational variables $\ell_{i,j}$ and $r_{i,j}$.
Intuitively,
$\ell_{i,j}$ (or $r_{i,j}$) will be a fraction of votes in which the
$j$th candidate is ranked at position $i$
and it is the most preferred (or the least preferred)
candidate among candidates from the subtree of $\calT$
rooted in the parent of this candidate, i.e.,
all candidates from $j$th to the $m$th one.
Now, consider the following linear programming constraints
(since we are only interested in the existence of a solution,
we give no objective function):
\begin{align}[left=\empheqlbrace]
    \notag & \ell_{i,j} + r_{i,j} = x_{i,j} \\
    \label{eq:prop:caterpillar:fixed:freq:p:cond1}
    & \quad \quad \quad \quad \mbox{for every } j, i \in [m]; \\
    \notag & \ell_{i-1,j} + r_{i+m-j,j} = \ell_{i,j+1} + r_{i+m-j-1,j+1}; \\
    \label{eq:prop:caterpillar:fixed:freq:p:cond2}
    & \quad \quad \quad \quad \mbox{for every } j \in [m-1] \mbox{ and } i \in
    [-m,m]; \\
    \notag & \ell_{i,j}, r_{i,j} \ge 0, \\
    \label{eq:prop:caterpillar:fixed:freq:p:cond3}
    & \quad \quad \quad \quad \mbox{for every } j \in [m] \mbox{ and } i \in
    [m];
\end{align}
where $\ell_{i,j}=r_{i,j}=0$, for each $j \in [m]$ and $i \in \mathbb{Z} \setminus [m]$.
In what follows, we will show the following claim:
\begin{claim}
\label{claim:caterpillar:fixed:freq:p}
For a given frequency matrix $X$,
there exists a caterpillar group-separable election $\elct$ realizing $X$,
if and only if,
there exist rational variables $\ell_{i,j},r_{i,j}$, for every $i,j \in [m]$,
satisfying conditions~\eqref{eq:prop:caterpillar:fixed:freq:p:cond1}--\eqref{eq:prop:caterpillar:fixed:freq:p:cond3}
\end{claim}

We will consider candidates $\{c_1,\dots,c_m\}$
and assume that they appear in the natural order
in the caterpillar tree $\calT$
(see \cref{fig:prop:caterpillar:fixed:freq:p} for an illustration).

For every $j \in [0,m-1]$ and $i \in [1,j+1]$,
by $S_{i,j}(\elct)$ let us denote the subset of votes in $\elct$ such that
candidates $c_{j+1},\dots,c_{m}$ occupy positions $[i,i+m-j-1]$
(not necessarily in this order).
Formally,
\[
    S_{i,j}(\elct) = \{ v \in V: \forall_{k \in [j+1,m]} \pos_v(c_k) \!\in\! [i,i+m-j-1] \}.
\]
We will skip $\elct$ for brevity.
We also add also set $S_{\cdot,m} = V$ and denote the collection of all such sets by $\mathcal{S}$.

Now,
for a given caterpillar tree $\calT$,
we will construct a directed graph, $G(\calT)$,
with the nodes from $\mathcal{S}$
(here we do not treat them as sets,
but as classes of votes;
hence, although technically,
for some $\elct$
some of the sets in $\mathcal{S}$ can be equal,
for the purpose of our graph construction
we will distinguish each of them as a separate node).
We illustrate our construction in Figure~\ref{fig:prop:caterpillar:fixed:freq:p}.

For each $j \in [m-1]$ and $i \in [m]$,
let us denote the set of voters $L(\elct)_{i,j}$ such that
every voter in $L(\elct)_{i,j}$ ranks candidate $c_j$ at position $i$
and prefers $c_j$ over $c_{j+1}$.
Formally,
\[
    L(\elct)_{i,j}= \{ v \in V : \pos_v(c_j)=i, c_j \succ_v c_{j+1} \}.
\]
If $\elct$ is a caterpillar group-separable election, $j \in [m-2]$, and $i \in [j+1]$,
then $v \in L(\elct)_{i,j}$ prefers $c_j$ over $c_k$, for all $k \in [j+1,m]$.
Thus, candidates $c_{j+1},\dots,c_{m}$ occupy positions $[i+1,i+m-j-1]$.
Furthermore, candidates $c_{j},c_{j+1},\dots,c_m$ occupy positions $[i,i+m-j-1]$.
Conversely, both facts imply that $c_j$ is at position $i$ and preferred over $c_{j+1}$,
thus
\begin{equation}
    \label{eq:prop:caterpillar:fixed:freq:p:l}
    L(\elct)_{i,j} = S_{i,j-1} \cap S_{i+1,j},
\end{equation}
for each $j \in [m-2]$ and $i \in [j+1]$.
Hence, let us add to graph $G(\calT)$ an arc from $ S_{i,j-1}$ to  $S_{i+1,j}$ that
corresponds to set $L(\elct)_{i,j}$.

Analogously, for each $j \in [m-1]$ and $i \in [m]$,
let us denote the set of voters $R(\elct)_{i,j}$ such that
every voter in $R(\elct)_{i,j}$ ranks candidate $c_j$ at position $i$,
but prefers $c_{j+1}$ over $c_{j}$, i.e.,
\[
    R(\elct)_{i+m-j,j}= \{ v \in V : \pos_v(c_j)=i, c_{j+1} \succ_v c_j \}.
\]
If $\elct$ is a caterpillar group-separable election, $j \in [m-2]$, and $i \in [j+1]$,
then $v \in R(\elct)_{i+m-j,j}$ prefers $c_k$ over $c_j$ for all $k \in [j+1,m]$.
Hence, candidates $c_{j+1},\dots,c_{m}$ occupy position $[i,i+m-j-1]$ and
candidates $c_{j},c_{j+1},\dots,c_m$ positions $[i,i+m-j]$.
Since, conversely, both facts imply that $c_j$ is at position $i+m-j$ and $c_{j+1}$ preferred over it,
we get that
\begin{equation}
    \label{eq:prop:caterpillar:fixed:freq:p:r}
    R(\elct)_{i+m-j,j} = S_{i,j-1} \cap S_{i,j},
\end{equation}
for each $j \in [m-2]$, and $i \in [j+1]$.
Thus, let us add to graph $G(\calT)$ an arc from $ S_{i,j-1}$ to  $S_{i,j}$ that
corresponds to set $R(\elct)_{i+m-j,j}$.

In sets $S_{1,m-1},\dots,S_{m,m-1}$,
the position of candidate $c_m$ is uniquely determined.
Hence, let us add to $G$ and arc from each of these sets to $S_{\cdot,m}$.

If $\elct$ is a caterpillar group-separable election,
then the fact that in vote $v$ candidates $c_j,\dots,c_m$ occupy positions $[i,i+m-j-1]$,
implies that candidate $c_j$ has to be either at position $i$ (and preferred over $c_{j+1}$)
or at position $i+m-j-1$ (and $c_{j+1}$ is preferred over it).
Hence,
\(
    S_{i,j} \subseteq L(\elct)_{i,j+1} \cup R(\elct)_{i+m-j-1,j+1}. 
\)
On the other hand,
from equations~\eqref{eq:prop:caterpillar:fixed:freq:p:l}
and~\eqref{eq:prop:caterpillar:fixed:freq:p:r}
we get that
\(
    S_{i,j} \supseteq L(\elct)_{i,j+1} \cup R(\elct)_{i+m-j-1,j+1}.
\)
Thus,
\begin{equation}
\label{eq:prop:caterpillar:fixed:freq:p:eq1}
    S_{i,j} = L(\elct)_{i,j+1} \cup R(\elct)_{i+m-j-1,j+1},    
\end{equation}
for every $j \in [0,m-1]$ and $i \in [1,j+1]$.
Moreover,
in such vote $v$,
candidate $c_{j-1}$ has to be either at position $i-1$ (and preferred over $c_{j}$)
or at position $i+m-j$ (and $c_{j}$ is preferred over it).
Thus, analogously we get that
\begin{equation}
\label{eq:prop:caterpillar:fixed:freq:p:eq2}
    S_{i,j} = L(\elct)_{i-1,j} \cup R(\elct)_{i+m-j,j},
\end{equation}
for every $j \in [1,m-1]$ and $i \in [1,j+1]$.

Building upon our construction of graph $G(\calT)$,
let us prove~\cref{claim:caterpillar:fixed:freq:p}.
Let us start by showing that if $X$ is a frequency matrix
realizable by some caterpillar group-separable election $\elct=(C,V)$,
then there exist rational constants $\ell_{i,j},r_{i,j}$, for every $i,j \in [m]$,
satisfying conditions~\eqref{eq:prop:caterpillar:fixed:freq:p:cond1}--\eqref{eq:prop:caterpillar:fixed:freq:p:cond3}.
To this end,
let us set $\ell_{i,j} = |L(\elct)_{i,j}|/|V|$ and $r_{i,j} = |R(\elct)_{i,j}|/|V|$,
for every $j \in [m-1]$ and $i \in [m]$, and
$\ell_{i,m} = x_{i,m}$ and $r_{i,m} = 0$,
for every $i \in [m]$.
Observe that all defined variables are indeed rational.

Now, observe that for every $j \in [m-1]$, in every $v \in V$,
candidate $c_j$ is either preferred to $c_{j+1}$ or
$c_{j+1}$ is preferred over $c_j$.
This means that $L(\elct)_{i,j}$ and $R(\elct)_{i,j}$ are disjoint
and $L(\elct)_{i,j} \cup R(\elct)_{i,j}$ is a set of exactly all votes in
which $c_j$ is ranked at position $i$.
Hence, we get that
\(
    \ell_{i,j} + r_{i,j} = x_{i,j},
\)
for every $j \in [m-1]$ and $i \in [m]$.
For $j=m$, we obtain the same equation
directly from our choice of $\ell_{i,m}$ and $r_{i,m}$.
Hence, condition~\eqref{eq:prop:caterpillar:fixed:freq:p:cond1} is satisfied.

Next, let us fix $j \in [m-2]$.
From equations~\eqref{eq:prop:caterpillar:fixed:freq:p:eq1}
and~\eqref{eq:prop:caterpillar:fixed:freq:p:eq2},
we get that
\[
    L(\elct)_{i-1,j} \cup R(\elct)_{i+m-j,j} = L(\elct)_{i,j+1} \cup R(\elct)_{i+m-j-1,j+1}.
\]
Since $i < i + m - j -1$,
sets $L(\elct)_{i,j+1}$ and $R(\elct)_{i+m-j-1,j+1}$
are disjoint.
The same is true for sets $L(\elct)_{i-1,j}$ and $R(\elct)_{i+m-j,j}$.
Thus, we obtain
\begin{multline*}
    |L(\elct)_{i-1,j}| + |R(\elct)_{i+m-j,j}| = \\
    |L(\elct)_{i,j+1}| + |R(\elct)_{i+m-j-1,j+1}|.
\end{multline*}
Dividing both sides by $|V|$ yields
\[
    \ell_{i-1,j} + r_{i+m-j,j} = \ell_{i,j+1} + r_{i+m-j-1,j+1}.
\]
For $j = m-1$, observe that we have
\[
     |L(\elct)_{i-1,m-1}| + |R(\elct)_{i+1,m-1}| = |S_{i,m-1}|,
\]
for every $i \in [m]$.
Now, observe that set $S_{i,m-1}$ is a set of exactly these votes
in which candidate $c_m$ is at position $i$.
Thus, we get
\[
     \ell_{i-1,m-1} + r_{i+1,m-1} = x_{i,m}.
\]
And since we took $\ell_{i,m} = x_{i,m}$ and $r_{i,m} = 0$, we get
\[
     \ell_{i-1,m-1} + r_{i+1,m-1} = \ell_{i,m} + r_{i,m}.
\]
Hence, condition~\eqref{eq:prop:caterpillar:fixed:freq:p:cond2} is satisfied.

Finally, since the cardinality of sets is always nonnegative,
condition~\eqref{eq:prop:caterpillar:fixed:freq:p:cond3} is also satisfied.

\begin{figure}[t]
    \centering
    \begin{tikzpicture}
      \tikzset{
        node/.style={circle,draw,minimum size=0.65cm,inner sep=0, fill = black!05, font=\scriptsize},
        arc/.style={->,draw,thick,-latex,above,font=\scriptsize},
        blank/.style={minimum size=0.001cm},
        edge/.style={-,draw,thick}
      }
      \def\sx{0.4cm} 
      \def\sy{0.6cm} 
      \def\x{2cm}
      \def\y{1cm}
      \node[blank] (T1) at (\x + 0*\sx, 0*\sy + \y) {};
      \node[blank] (T2) at (\x + 1*\sx, -1*\sy + \y) {};
      \node[blank] (T3) at (\x + 2*\sx, -2*\sy + \y) {};
      \node[blank] (T4) at (\x + 3*\sx, -3*\sy + \y) {};
      \node[blank] (Ta) at (\x + -1*\sx, -1*\sy + \y) {$c_1$};
      \node[blank] (Tb) at (\x + 0*\sx, -2*\sy + \y) {$c_2$};
      \node[blank] (Tc) at (\x + 1*\sx, -3*\sy + \y) {$c_3$};
      \node[blank] (Td) at (\x + 2*\sx, -4*\sy + \y) {$c_4$};
      \node[blank] (Te) at (\x + 4*\sx, -4*\sy + \y) {$c_5$};
      \node[blank] (T_name) at (\x + 2.5*\sx, -1*\sy + \y) {$\mathcal{T}$};
      
      \path[edge]
      (T1) edge (T2)
      (T2) edge (T3)
      (T3) edge (T4)
      (T1) edge (Ta)
      (T2) edge (Tb)
      (T3) edge (Tc)
      (T4) edge (Td)
      (T4) edge (Te)
      ;

      \def\sx{1.6cm} 
      \def\sy{1.2cm} 
      \def\x{0cm}
      \def\y{0cm}
      \node[node] (10) at (\x + 0*\sx, 0*\sy + \y) {$S_{1,0}$};
      \node[node] (11) at (\x + -0.5*\sx, -1*\sy + \y) {$S_{2,1}$};
      \node[node] (21) at (\x +  0.5*\sx, -1*\sy + \y) {$S_{1,1}$};
      \node[node] (12) at (\x + -1*\sx, -2*\sy + \y) {$S_{3,2}$};
      \node[node] (22) at (\x +  0*\sx, -2*\sy + \y) {$S_{2,2}$};
      \node[node] (32) at (\x +  1*\sx, -2*\sy + \y) {$S_{1,2}$};
      \node[node] (13) at (\x + -1.5*\sx, -3*\sy + \y) {$S_{4,3}$};
      \node[node] (23) at (\x + -0.5*\sx, -3*\sy + \y) {$S_{3,3}$};
      \node[node] (33) at (\x +  0.5*\sx, -3*\sy + \y) {$S_{2,3}$};
      \node[node] (43) at (\x +  1.5*\sx, -3*\sy + \y) {$S_{1,3}$};
      \node[node] (14) at (\x + -2*\sx, -4*\sy + \y) {$S_{5,4}$};
      \node[node] (24) at (\x + -1*\sx, -4*\sy + \y) {$S_{4,4}$};
      \node[node] (34) at (\x +  0*\sx, -4*\sy + \y) {$S_{3,4}$};
      \node[node] (44) at (\x +  1*\sx, -4*\sy + \y) {$S_{2,4}$};
      \node[node] (54) at (\x +  2*\sx, -4*\sy + \y) {$S_{1,4}$};
      \node[node] (_5) at (\x +  0*\sx, -5*\sy + \y) {$S_{\cdot,m}$};
      \node[blank] (G_name) at (\x + -1.5*\sx, -1*\sy + \y) {$G(\calT)$};
      
      \path[arc]
      (10) edge node[left = 0.17cm,pos = 0.8,anchor=south] {$\ell_{1,1}$} (11)
      (10) edge node[left = 0.45cm,pos = 0.8,anchor=south] {$r_{5,1}$} (21)
      (11) edge node[left = 0.17cm,pos = 0.8,anchor=south] {$\ell_{2,2}$} (12)
      (11) edge node[left = 0.45cm,pos = 0.8,anchor=south] {$r_{5,2}$} (22)
      (21) edge node[left = 0.17cm,pos = 0.8,anchor=south] {$\ell_{1,2}$} (22)
      (21) edge node[left = 0.45cm,pos = 0.8,anchor=south] {$r_{4,2}$} (32)
      (12) edge node[left = 0.17cm,pos = 0.8,anchor=south] {$\ell_{3,3}$} (13)
      (12) edge node[left = 0.45cm,pos = 0.8,anchor=south] {$r_{5,3}$} (23)
      (22) edge node[left = 0.17cm,pos = 0.8,anchor=south] {$\ell_{2,3}$} (23)
      (22) edge node[left = 0.45cm,pos = 0.8,anchor=south] {$r_{4,3}$} (33)
      (32) edge node[left = 0.17cm,pos = 0.8,anchor=south] {$\ell_{1,3}$} (33)
      (32) edge node[left = 0.45cm,pos = 0.8,anchor=south] {$r_{3,3}$} (43)
      (13) edge node[left = 0.17cm,pos = 0.8,anchor=south] {$\ell_{4,4}$} (14)
      (13) edge node[left = 0.45cm,pos = 0.8,anchor=south] {$r_{5,4}$} (24)
      (23) edge node[left = 0.17cm,pos = 0.8,anchor=south] {$\ell_{3,4}$} (24)
      (23) edge node[left = 0.45cm,pos = 0.8,anchor=south] {$r_{4,4}$} (34)
      (33) edge node[left = 0.17cm,pos = 0.8,anchor=south] {$\ell_{2,4}$} (34)
      (33) edge node[left = 0.45cm,pos = 0.8,anchor=south] {$r_{3,4}$} (44)
      (43) edge node[left = 0.17cm,pos = 0.8,anchor=south] {$\ell_{1,4}$} (44)
      (43) edge node[left = 0.45cm,pos = 0.8,anchor=south] {$r_{2,4}$} (54)
      (14) edge[bend left =-20] node[left = 0.8cm,pos = 0.38,anchor=south] {$x_{5,5}$} (_5)
      (24) edge[bend left =-10] node[left = 0.65cm,pos = 0.7,anchor=south] {$x_{4,5}$} (_5)
      (34) edge[bend left =  0] node[left = 0.3cm,pos = 0.95,anchor=south] {$x_{3,5}$} (_5)
      (44) edge[bend left = 10] node[left = 0.10cm,pos = 0.7,anchor=south] {$x_{2,5}$} (_5)
      (54) edge[bend left = 20] node[left = 0.05cm,pos = 0.4,anchor=south] {$x_{1,5}$} (_5)
      ;
    \end{tikzpicture}
    \caption{A tree of a caterpillar group-separable election with 5 candidates ($\mathcal{T}$)
    and a graph constructed for such election as in proof of \cref{prop:caterpillar:fixed:freq:p} ($G$).}
    \label{fig:prop:caterpillar:fixed:freq:p}
\end{figure}
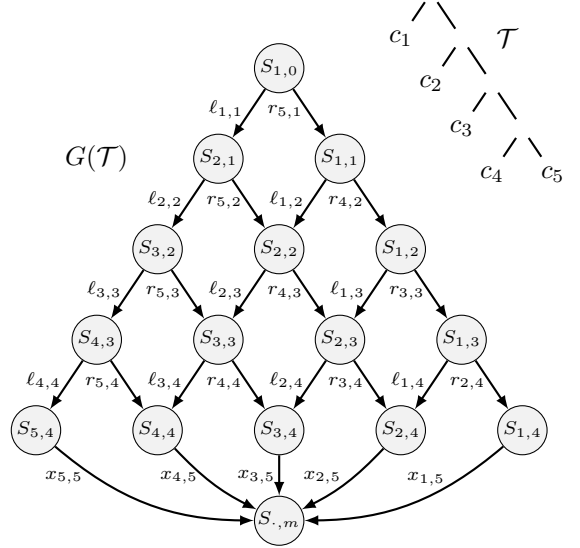

In the remainder of the proof,
let us show that if for a given frequency matrix $X$
there exist rational constants $\ell_{i,j},r_{i,j}$, for every $i,j \in [m]$,
satisfying conditions~\eqref{eq:prop:caterpillar:fixed:freq:p:cond1}--\eqref{eq:prop:caterpillar:fixed:freq:p:cond3},
then there exists a caterpillar group-separable election $\elct$
that realizes $X$.

To this end,
observe that for every $j \in [0,m-1]$ and $i \in [j+1]$,
each edge outgoing from $S_{i,j}$ in graph $G(\calT)$
corresponds to placement of candidate $c_j$ at particular position.
Hence, each path from $S_{1,0}$ to $S_{\cdot,m}$ in graph $G(\calT)$ 
corresponds to placements of all candidates $c_1,\dots,c_m$, i.e.,
a vote.
Moreover, each such vote can be present in a caterpillar group-separable election
compatible with tree $\calT$.
Thus, if all votes in an election, $\elct$, would be obtained in this way,
then $\elct$ would be a caterpillar group-separable election compatible with $\calT$.
We will show that it is possible,
by considering flows from $S_{1,0}$ to $S_{\cdot,m}$ on graph $G(\calT)$.


Since for every $i,j \in [m]$ constants $\ell_{i,j}$ and $r_{i,j}$ are rational,
there exists $n \in \mathbb{N}$ such that $\ell_{i,j} \cdot n$ and $r_{i,j} \cdot n$ are
integers for every $i,j \in [m]$.
Let us define \emph{capacity} function, $c : E(G(\calT)) \rightarrow \mathbb{R}_{\ge 0}$,
that for each arc returns the maximal number of flows that can go through this arc.
Specifically, we set
\begin{align*}[left=\empheqlbrace]
    & c(S_{i,j},S_{i+1,j+1}) = \ell_{i,j+1} \cdot n, \\
    & \quad \quad \quad \quad \mbox{for all } j \in [0,m-2], i \in [1,j+1], \\
    & c(S_{i,j},S_{i,j+1}) = r_{i + m - j - 1,j+1} \cdot n, \\
    & \quad \quad \quad \quad \mbox{for all } j \in [0,m-2], i \in [1,j+1], \\
    & c(S_{i,m-1}) = x_{i,m} \cdot n, \\
    & \quad \quad \quad \quad \mbox{for all } i \in [m]
\end{align*}
(see \cref{fig:prop:caterpillar:fixed:freq:p} for an illustration).

Observe that the sum of capacities of arcs incoming to $S_{\cdot,m}$
is equal to $n$.
Condition~\eqref{eq:prop:caterpillar:fixed:freq:p:cond2} ensures that
the summed capacity of arcs incoming to $S_{i,j}$ is equal to
the summed capacity of arcs outgoing from $S_{i,j}$,
for every $j \in [m-1]$ and $i \in [j+1]$.
Therefore, there exist a (possibly multi-) set of $n$ paths on $G$ starting in
$S_{1,0}$ and ending in $S_{\cdot,m}$,
such that each edge $e \in E(G(\calT))$ is traversed by exactly $c(e)$ paths.
Each such path can be identified with a particular preference order of candidates in $C$.
Let us denote all $n$ of these preference orders by $V$.
Then election $\elct = (C,V)$ is a caterpillar group-separable election
compatible with $\calT$.
Furthermore, from condition~\eqref{eq:prop:caterpillar:fixed:freq:p:cond1}
we know that $\elct$ realizes matrix $X$,
which concludes the proof.
\end{proof}

Due to a certain relation between caterpillar group-separable and
single-peaked elections, the above results immediately implies the
next one (in short, a position matrix of a single-peaked election is a
transposition of a position matrix of a caterpillar group-separable
one).

\begin{proposition}\label{prop:single-peaked:fixed:freq:p}
    There is a polynomial-time algorithm
    that given a frequency matrix $X$
    and a single-peaked domain $\calD$,
    decides whether there is an election that realizes $X$,
    and whose all votes belong to $\calD$.
\end{proposition}
\begin{proof}
\citet{boe-bre-elk-fal-szu:c:frequency-matrices} have shown
a one-to-one correspondence between
single-peaked elections with a given societal axis
and caterpillar group-separable elections compatible with a given tree.
Moreover, a transposition frequency matrix of a single-peaked election
is the frequency matrix of the corresponding caterpillar group-separable election.
Therefore,
to check whether a frequency matrix $X$
can be realized by a single-peaked election,
we can transpose it
and use an algorithm described in \cref{prop:caterpillar:fixed:freq:p}
to check if it is realizable by the corresponding caterpillar group-separable election.
\end{proof}

\crefrange{prop:input-given-p}{prop:single-peaked:fixed:freq:p}
combined imply the claim of the theorem, which concludes the proof.
\end{proof}

\subsection[Proof of Theorem 4]{Proof of \cref{thm:balanced}}

\balanced*
\begin{proof}
The algorithm for checking if a given position matrix can be realized
by a balanced group-separable election $\elct$
is given as \cref{alg:realizability_by_balanced}.
The algorithm is based on the condition from the proof of \cref{prop:balanced:fixed:p}
that a frequency matrix $X$
can be realized by a balanced group-separable election
compatible with a tree in which candidates on consecutive leaves
correspond to consecutive columns of $X$,
if and only if,
$X$ is maximally-evenly-quartered
(see the proof of \cref{prop:balanced:fixed:p} for the definition
of a maximally-evenly-quartered matrix).
The algorithm checks if it is possible
to change the order of columns of $X$ in such a way
that it is maximally-evenly-quartered.
Throughout this proof,
we will identify candidates with the columns of
matrix $X$, i.e., we set $C = [m]$.

Let us first describe the algorithm
and then prove its correctness. 
The algorithm works recursively.
We start by checking the borderline case in which
matrix $X$ is a one-element matrix (lines \ref{alg:trivbeg}--\ref{alg:trivend}).
If it is so,
then we return True,
as $X$ can be realized by a trivial
election with one candidate; clearly, such election is balanced group-separable.
If $X$ is $m \times m$ matrix for $m = 2^k$, where $k>0$,
we list all possible pairs of candidates
in the PossibleSiblings set (lines~\ref{alg:defpossiblingsbeg}--\ref{alg:defpossiblingsend}).
In this set,
we store the pairs of candidates
that can end up as sibling leaves in the 
tree of balanced group-separable election
that realizes $X$.
If two candidates, $j$ and $j'$,
are sibling leaves in the 
tree of balanced group-separable election,
then for every $i \in [m/2]$ it must hold
that
\begin{equation}
    \label{eq:prop:balanced:p}
    x_{2i,j}=x_{2i-1,j'}
    \quad \mbox{and} \quad
    x_{2i,j'}=x_{2i-1,j}.
\end{equation}
This is because, matrix
\[
    Q_{j,j'} = \left[
    \begin{array}{cc}
       x_{2i-1,j}  & x_{2i-1,j'} \\
       x_{2i,j}  & x_{2i,j'}
    \end{array}
    \right]
\]
can be obtained by taking the quarters from $X$
when we reorder its columns
in the order of leaves in the candidate tree.
Hence, if such reordered $X$ is maximally-evenly-quartered,
then matrix $Q_{j,j'}$ has to be evenly-quartered,
which is equivalent to equation~\eqref{eq:prop:balanced:p}.
Thus, for every pair of candidates and every $i \in [m/2]$
we check if equation~\eqref{eq:prop:balanced:p} holds.
If for some pair $\{j,j'\} \in C$
and some $i \in [m/2]$ it does not,
we remove this pair from PossibleSiblings (lines~\ref{alg:processsiblingsbeg}--\ref{alg:processsiblingsend}).

Next, we check if there is a perfect matching in the graph of candidates
with the set of edges that is the final PossibleSiblings set.
If the perfect matching exists,
we store it in list $M$ (lines \ref{alg:pmsearchbeg}--\ref{alg:pmsearchend}).
A perfect matching can be found greedily,
i.e., we can match each candidate to
the first unmatched candidate it is connected to.
To see why,
observe that if candidate~$j$ is connected to
two distinct candidates $c$ and $c'$, i.e.,
$\{j,c\}, \{j,c'\} \in$ PossibleSiblings,
then equation~\eqref{eq:prop:balanced:p}
implies that $x_{i,c}=x_{i,c'}$,
for every $i \in [m]$.
Hence, columns~$c$ and~$c'$ are identical and
matching $j$ with $c$ or with $c'$ is equivalent.

If there is no perfect matching,
this means that it is impossible to pair candidates to obtain
the lowest level of binary tree of candidates.
Hence, we return False (lines~\ref{alg:nosolutionbeg}--\ref{alg:nosolutionbeg}).

If there is a perfect matching of candidates that can be siblings,
then we construct $m/2 \times m/2$ matrix $Y$ in which 
each column $j$ corresponds to, each time different,
matching $\{c,c'\} = M[j]$, in an arbitrary order
(lines~\ref{alg:pmprocessbeg}--\ref{alg:pmprocessend}).
For each $i,j \in [m/2]$ and $\{c,c'\} = M[j]$,
we set entry $Y[i,j]$ of matrix $Y$ to
the sum of entries in matrix $Q_{c,c'}$ divided by 2,
which, by equation~\eqref{eq:prop:balanced:p}, is equal to
\[
  x_{2i,c} + x_{2i,c'}.
\]
This way, matrix $Y$ is also a position matrix.

Finally, we repeat the algorithm, this time for matrix $Y$
(line~\ref{alg:repeatfory}).

\renewcommand{\algorithmicrequire}{\textbf{Input:}}
\renewcommand{\algorithmicensure}{\textbf{Output:}}
\begin{algorithm}[t]
\caption{RealizabilityByBalanced}\label{alg:realizability_by_balanced}
\begin{algorithmic}[1]
\REQUIRE $m \times m$ position matrix $X$, where $m = 2^k$ for some $k \in \mathbb{N}$
\ENSURE Does there exist a balanced group-separable election, $\elct$, that realizes $X$?
\IF {$m=1$} \label{alg:trivbeg}
  \RETURN True
\ENDIF \label{alg:trivend}
\STATE $C\gets [m]$ \label{alg:defpossiblingsbeg}
\STATE PossibleSiblings $\gets \{ S \subseteq C: |S| = 2\}$
\label{alg:defpossiblingsend}
\FOR {$i$ \textbf{in} $[m/2]$} \label{alg:processsiblingsbeg}
    \FOR {$\{j,j'\}$ \textbf{in} PossibleSiblings}
        \IF {\textbf{not}($x_{2i,j}=x_{2i-1,j'}$ \textbf{and} $x_{2i-1,j}=x_{2i,j'}$)}
            \STATE PossibleSiblings$\gets $PossibleSiblings$ \setminus \{\{j,j'\}\}$
        \ENDIF
    \ENDFOR
\ENDFOR \label{alg:processsiblingsend}
\IF {HasPerfectMatching$(C,$PossibleSiblings$)$} \label{alg:pmsearchbeg}
    \STATE $M \gets$ PerfectMatching$(C,$PossibleSiblings$)$
    \label{alg:pmsearchend}
    \STATE $Y \gets m/2 \times m/2$ matrix \label{alg:pmprocessbeg}
    \FOR {$i$ \textbf{in} $[m/2]$}
        \FOR {$j$ \textbf{in} $[m/2]$}
            \STATE $\{c,c'\} \gets M[j]$
            \STATE $Y[i,j] \gets x_{2i,c}+x_{2i,c'}$
        \ENDFOR
    \ENDFOR \label{alg:pmprocessend}
    \RETURN RealizabilityByBalanced$(Y)$\label{alg:repeatfory}
\ELSE \label{alg:nosolutionbeg}
    \RETURN False \label{alg:nosolutionend}
\ENDIF
\end{algorithmic}
\end{algorithm}

Now, let us prove the correctness of our algorithm.
To this end, we will follow the induction by $k$,
i.e., the binary logarithm of the number of candidates.
If $k=0$, then our algorithm always returns True,
and matrix $X$ is always realizable by a trivial election
with one candidate,
hence the thesis holds.
Let us then assume that our thesis holds for all
$2^k \times 2^k$ position matrices
and let us consider an arbitrary
$2^{k+1} \times 2^{k+1}$ position matrix $X$.

First, let us show that if there exists
a balanced group-separable election $\elct = (C,V)$ that realizes $X$,
then our algorithm returns True.
Let $\calT$ be a binary tree with which $\elct$ is compatible.
Observe that the result of \cref{alg:realizability_by_balanced}
does not depend on the order of the columns of the input position matrix.
Thus, without loss of generality, we can assume that
the order of the candidates at the leaves of $\calT$
aligns with the order of corresponding columns in matrix $X$
(otherwise the columns of $X$ can be reordered).
Let us denote the candidates in $\elct$ as $C = [m]$,
where the order $1,\dots,m$ is the order
in which they appear at the leaves of $\calT$.
Then, our algorithm will assign to $M$
a perfect matching of candidates
that is $(\{1,2\},\{3,4\},\dots,\{m-1,m\})$
(or equivalent if some columns in $X$ are identical).
Consequently, matrix $Y$ will be a $2^k \times 2^k$ matrix
in which every entry is equal to the sum of four neighboring
elements of matrix $X$ divided by 2, i.e.,
\[
    y_{i,j} = x_{2i,2j} + x_{2i,2j-1},
\]
for each $i,j \in [m/2]$.
Now, the outcome of our algorithm for matrix $X$
will be the same as the outcome for matrix $Y$.
Hence, we have to show that the outcome for $Y$ is True.
By the inductive assumption, it suffices to
construct balanced group-separable election that
realizes $Y$.
To this end, let us denote $C' = \{\{1,2\},\{3,4\},\dots,\{m-1,m\}\}$,
i.e., the candidates $C'$ are matched pairs of candidates from $C$.
Since $\elct$ is a balanced group-separable election,
for every voter $v \in V$ and
two distinct pairs of candidates $\{j,j+1\},\{j',j'+1\} \in C'$,
it holds that either $v$ prefers both $j$ and $j+1$ over both $j'$ and $j'+1$
or the converse is true, i.e.,
$v$ prefers both $j'$ and $j'+1$ over both $j$ and $j+1$.
Thus, we can define an \emph{aggregated vote} $v$, denoted by $f(v)$,
as a preference order on $C'$ in which
\[
    \begin{cases}
        \{j,j+1\} \succ_{f(v)} \{j',j'+1\}, & \mbox{if } j \succ_v j', \\
        \{j',j'+1\} \succ_{f(v)} \{j,j + 1\}, & \mbox{if } j' \succ_v j,
    \end{cases}
\]
for every $\{j,j+1\},\{j',j'+1\} \in C'$.
Furthermore, let
\[
    V' = \{ f(v) : v \in V \}.
\]
Then, election $\elct' = (C',V')$ is also a balanced group-separable election
compatible with tree $\calT$ with its leaves removed.
Observe that for every $i,j \in [m/2]$
candidate $\{2j-1,2j\}$ is ranked at position $i$
by voter $f(v) \in V'$,
if and only if,
either candidate $2j-1$ or candidate $2j$ is ranked at position $2i$
by voter $v$.
Hence, election $\elct'$ realizes matrix $Y$.
Therefore, by the inductive assumption,
our algorithm returns True for~$Y$,
which means that it also returns True for~$X$.

Finally, let us prove that if our algorithm returns True,
then there exists
a balanced group-separable election that realizes $X$.
Since we return True and $k>0$,
there exists a perfect matching~$M$ of candidates
such that for every $\{c,c'\} \in M$ and every $i \in [m/2]$
equation~\eqref{eq:prop:balanced:p} holds.
Without loss of generality,
let us assume that $M = \{\{1,2\},\{3,4\},\dots,\{m-1,m\}\}$
(otherwise, we can reorder the columns of matrix $X$).
Thus, we can construct matrix $Y$ as in the algorithm.
Then, from the inductive assumption
we know that there exists
a balanced group-separable election,
$\elct' = (C',V')$,
that realizes $Y$
in which $C' = M$.

Now, based on election $\elct'$,
let us construct election $\elct'' = (C,V'')$
in which we exchange each candidate $\{c,c'\}$
for a pair for candidates $c,c'$.
Specifically,
for a preference order $v$ on candidates $C'$,
let us define a \emph{disaggregated vote} $v$, denoted by $g(v)$,
as a preference order on $C$ in which
\(
    c \succ_{g(v)} d,
\)
for every $c,d \in C$ such that $\{c,d\} \not \in C'$ and $\{c,c'\} \succ_v \{d,d'\}$ for some $c',d' \in C$,
and $2j-1 \succ_{g(v)} 2j$ for every $j \in [m/2]$.
Then, we set 
\(
    V'' = \{ g(v) : v \in V' \}.
\)
Observe that election $\elct''$ is still
balanced group-separable election and it
realizes position matrix
\begin{small}
\[
\arraycolsep=1.4pt\def\arraystretch{1.4}
    Y' = \left[
    \begin{array}{ccccccc}
       y_{1,1} & 0  &  y_{1,2} & 0  &       & y_{1,\frac{m}{2}} & 0 \\
       0 & y_{1,1}  &  0 & y_{1,2}  &\cdots & 0 & y_{1,\frac{m}{2}} \\
       y_{2,1} & 0  &  y_{2,2} & 0  &       & y_{2,\frac{m}{2}} & 0 \\
       0 & y_{2,1}  &  0 & y_{2,2}  &       & 0 & y_{2,\frac{m}{2}} \\
         &  \vdots  &    &          &\ddots &   & \vdots    \\
       y_{\frac{m}{2},1} & 0  &  y_{\frac{m}{2},2} & 0  &       & y_{\frac{m}{2},\frac{m}{2}} & 0 \\
       0 & y_{\frac{m}{2},1}  &  0 & y_{\frac{m}{2},2}  &\cdots & 0 & y_{\frac{m}{2},\frac{m}{2}}
    \end{array}
    \right].
\]
\end{small}

In what follows,
we will construct a sequence of elections,
$\elct_0,\elct_1,\dots,\elct_{m/2}$,
such that $\elct_0 = \elct'$
and $\elct_{m/2}$ will realize matrix $X$.
For each $j \in [m/2]$,
we will construct election $\elct_j$ from $\elct_{j-1}$,
by swapping the positions of candidates $2j-1$ and $2j$
in some of the votes
(observe that they are always at consecutive positions).
Let us describe
how we obtain
election $\elct_j$ from $\elct_{j-1}$
in more detail.
Let us denote $\elct_{j-1}=(C,V_{j-1})$.
Next, let us split set $V_{j-1}$ into $m/2$ sets
$V^{2}_{j-1},V^4_{j-1},\dots,V^m_{j-1}$
depending on the position of candidate $2j$ in a vote.
Specifically,
for each $i \in [m/2]$, let
\[
    V^{2i}_{j-1} = \{ v \in V_{j-1} : \pos_v(2j)=2i\}.
\]
Then, from each set $V^{2i}_{j-1}$
we construct set $U^{2i}_{j-1}$
by arbitrarily choosing $x_{2i-1,2j}$ votes
in which we swap the positions of candidates $2j$ and $2j-1$
(so, now $2j-1$ is ranked at position $2i$
and $2j$ at position $2i-1$).
Then, we set
\(
    V_j = \bigcup_{i = 1}^{[m/2]} U^{2i}_{j-1}
\)
and
$\elct_{j} = (C,V_j)$.
Since we always swap only the candidates within one pair,
$\elct_j$ is balanced group-separable election for each $j \in [m/2]$.
Moreover, for every $j \in [m/2]$,
in each election $\elct_{j'}$ for $j' \ge j$,
candidate $j$ is at position $i$
in exactly $x_{i,j}$ votes.
Hence, the position matrix of election
$\elct_{m/2}$ is indeed matrix $X$.
This concludes the proof.
\end{proof}

\section[Additional Material for Section 5]{Additional Material for \cref{sec:condorcet}}
In this section
we provide the proofs of \cref{thm:condorcet-nph,thm:condorcet:condition} and
additional details of experiments
concerning the Condorcet winners in elections
realizing a given position matrix.

\subsection[Proof of Theorem 6]{Proof of \cref{thm:condorcet-nph}}

\condorcetnph*
\begin{proof}
 \newcommand{\univ}{\ensuremath{\mathcal{U}}}
 \newcommand{\uelem}{\ensuremath{u}}
 \newcommand{\setcovsize}{\ensuremath{t}}
 \newcommand{\setsfam}{\ensuremath{\mathcal{S}}}
 \newcommand{\sset}{\ensuremath{S}}
 \newcommand{\domain}{\ensuremath{\mathcal{D}}}
 \newcommand{\posmtrx}{\ensuremath{X}}
 \newcommand{\advmatrix}{\ensuremath{A_1}}
 \newcommand{\constrmatrix}{\ensuremath{A_2}}
 \newcommand{\XTSC}{\textsc{Exact Cover by 3-Sets}}
 \newcommand{\originstance}{\ensuremath{\mathcal{I}}}
 \newcommand{\fininstance}{\ensuremath{\mathcal{I}'}}
 \newcommand{\ecand}{\ensuremath{e}}
 \newcommand{\fcand}{\ensuremath{f}}
 \newcommand{\pcand}{\ensuremath{p}}
 \newcommand{\dcand}{\ensuremath{d}}
 \newcommand{\ssetvotfam}[1]{\ensuremath{V({#1})}}
 \newcommand{\ssetel}{\ensuremath{s}}

 \newcommand{\diagval}{\ensuremath{h}}
 \newcommand{\erank}{\ensuremath{E}}

 \newcommand{\excov}{\ensuremath{K}}
 \newcommand{\reprcov}{\ensuremath{\textrm{repr}}}

 \newcommand{\voteforp}{\ensuremath{\voter_p}}
 \newcommand{\voteford}{\ensuremath{\voter_d}}

 \newcommand{\pwscmp}{\ensuremath{N}}
 \newcommand{\votefrompgroup}{\ensuremath{v^+}}
 \newcommand{\votefromdgroup}{\ensuremath{v^-}}

 \newcommand{\tinysucc}{\triangleright}

 \begingroup
 We reduce from \textsc{Exact Cover by 3-Sets} where given a
 $3\setcovsize$-element universe $\univ=\{\uelem_1, \uelem_2, \ldots,
 \uelem_{3\setcovsize}\}$ and a family of 3-element subsets
 $\setsfam = \{ S_1, S_2, \dots, S_k\}$, the
 question is whether there exists a $\setcovsize$-size set $\excov \subseteq
 \setsfam$ forming an exact cover of~$\univ$, that is, $\bigcup_{\sset \in \excov}
 \sset = \univ$. For convenience, we assume that elements of subsets~$\sset$ are
 ordered ascending. Furthermore, we call the smallest element of a
 subset~$\sset$ the \emph{representative of $\sset$}. Consequently, for some
 exact cover~\excov{}, we refer to the set of representatives of subsets
 belonging to \excov{} as to~$\reprcov(\excov)$. We note
 that~$\reprcov(\excov)$, consisting of exactly~$\setcovsize$ distinct elements,
 yields too little information to uniquely define~\excov{}.
 
 We transform an instance~\originstance{} of~\XTSC{} into
 instance~$\fininstance{} = (\cnds, \posmtrx, \domain, \pcand)$ of a problem in
 which we ask whether there exists a realization of~\posmtrx{} using only votes
 from~\domain{} such that~\pcand{} is the Condorcet winner. Our reduction
 guarantees that there is always at least one feasible realization of~\posmtrx{}
 using votes from~\domain{}.
 
 \subsubsection{Construction}
 Regarding the candidates of~\fininstance{}, for each element~$\uelem_i \in
 \univ$, we add three candidates~$\{\cand_i, \ecand_i, \fcand_i\}$ that we call,
 respectively, \emph{element}, \emph{executive}, and \emph{filler} candidates of
 element~$u_i$.  Furthermore, we add the \emph{preferred} candidate~\pcand{} and
 the~\emph{despised} candidate~$\dcand$. Overall, we have $9\setcovsize +
 2$~candidates.

 \setcounter{MaxMatrixCols}{30}
 For the sake of presentation, we define our matrix~\posmtrx{} (of
 size~$(9\setcovsize + 2) \times (9\setcovsize + 2)$) as the sum of two
 matrices~$\advmatrix$ and~$\constrmatrix$.
 For reasons of clarity, when presenting a matrix, we leave zero-entries
 blank.

 First, we present matrix~$\advmatrix$ describing $11\setcovsize$~votes,
 in which candidates $\pcand$ and $\dcand$ are on first and second position
 and each of the remaining candidates always occupy only one position:
 \advantageMatrixFullExpr{}
 
 Next, let us define matrix~$\constrmatrix$, which encodes the
 remaining $9\setcovsize$~votes.
 To this end, let us first present one
 $(3\setcovsize) \times (3\setcovsize  -1)$ block of this matrix
 denoted by $B$.
 \helperMatrixDFullExpr{}
 
 \noindent Building upon it, let us take~$\diagval \coloneqq 9\setcovsize - 1$
 and define matrix~$\constrmatrix$ as follows:
 \matrixOfConstrainsFullExpr{}

 \noindent Note that~$\posmtrx = \advmatrix + \constrmatrix$ describes exactly $20
 \setcovsize$~votes; indeed, it is easy to verify that entries in each row and
 in each column of~\posmtrx{} sum to~$20\setcovsize$.

 \let\oldsucc\succ
 \renewcommand{\succ}{\ensuremath{\raisebox{1pt}{\resizebox{1.5ex}{!}{$\,\oldsucc\,$}}}}
 
 Now, we characterize set $\domain$,
 from which the votes in election $\elct$ must come.
 For the reasons of clarity,
 we define~$\domain$ using two sets~$\domain_1$ and~$\domain_2$ such
 that~$\domain_1 \cup \domain_2 = \domain$. 
 We construct them in such a way,
 that matrix~$\advmatrix$ is realizable by an election restricted to $\domain_1$
 and matrix~$\constrmatrix$ by an election restricted to $\domain_2$.
 Taking the union of the votes in both elections,
 will give us an election realizing $\posmtrx$.

 \newcommand{\swapfunc}{\ensuremath{s}}
 Let us define the following helper partial votes:
 \begin{align*}
    FC \colon \quad & \fcand_1 \succ \cand_1 \succ \fcand_2 \succ \cand_2 \succ \dots
    \succ \fcand_{3\setcovsize} \succ \cand_{3\setcovsize}, \\
    E_{-i} \colon \quad & \ecand_1 \succ \ecand_2 \succ \dots 
    \succ \ecand_{i-1} \succ \ecand_{i+1} \succ \dots
    \ecand_{3\setcovsize}, \textrm{ for every } i \in [3\setcovsize].
 \end{align*}
 Moreover, for some vote~$v$ 
 and two distinct candidates $x_i$, $x_j$, by $s(v, x_i, x_j)$ we denote a vote
 emerging from copying~$v$ and then swapping the positions of~$x_i$ and~$x_j$
 with each other. Note that we sometimes chain function~$s$ to express a vote
 coming up from taking some original vote and making several swaps (we never
 swap a single voter many times). For example, $\swapfunc(\swapfunc(v, x_i,x_j),
 x_{i'}, x_{j'})$ means a vote built by copying~$v$ and then pairwise swapping $x_i$
 with~$x_j$ and $x_{i'}$ with~$x_{j'}$.

 Now, let us define $\domain_1 \coloneqq \{ \voteforp , \voteford \}$, where
 \begin{align*}
  \voteforp &\colon \quad \pcand \succ \dcand \succ e_1 \succ E_{-1} \succ FC \quad \mbox{and} \\
  \voteford &\colon \quad \dcand \succ \pcand \succ e_1 \succ E_{-1} \succ FC.
 \end{align*}

 \noindent
 To define~$\domain_2$, we first introduce the following six votes
 for each element~$\uelem_i \in \univ$.
 \begin{align*}
    \votefrompgroup_1(\uelem_i) &\colon \quad
        FC \succ p \succ d \succ e_i \succ E_{-i},\\
    \votefrompgroup_2(\uelem_i) &\colon \quad
        FC \succ d \succ e_i \succ p \succ E_{-i},\\
    \votefrompgroup_3(\uelem_i) &\colon \quad
        FC\succ e_i \succ p \succ d  \succ E_{-i},\\
    \votefromdgroup_1(\uelem_i) &\colon \quad
        FC \succ d \succ p \succ e_i \succ E_{-i},\\
    \votefromdgroup_2(\uelem_i) &\colon \quad
        FC \succ p \succ e_i \succ d \succ E_{-i}, \quad \mbox{and}\\
    \votefromdgroup_3(\uelem_i) &\colon \quad
        \swapfunc(FC, f_i, c_i)\succ e_i \succ d \succ p  \succ E_{-i}.
  \end{align*}
  Next, for each set $S_i = \{u_x, u_y, u_z\} \in \setsfam$
  (where $u_x$ is the representative of $S_i$),
  we introduce vote
  \[
    \votefrompgroup(\sset_i) \colon 
    \swapfunc(\swapfunc(\swapfunc(FC, c_x, f_x), c_y, f_y), c_z,
    f_z)
    \succ 
    \ecand_x \succ \pcand \succ \dcand  \succ E_{-x}.
  \]
 Finally, we set
 \begin{multline*}
    \domain_2 \coloneqq \{ \votefrompgroup_j(\uelem_i), \votefromdgroup_j(\uelem_i) : \uelem_i \in \univ, j \in [3]\} \ 
    \cup \\ \{ \votefrompgroup(\sset_i) : \sset_i \in \setsfam\}.
 \end{multline*}
 \noindent This concludes the construction of instance~$\fininstance$.

 \subsubsection{Correctness}
 First, we show that $\posmtrx$ is indeed
 always realizable by some election 
 restricted to domain $\domain$~(I).
 Then, we prove that if instance $\originstance$ has
 an exact cover, then $\pcand$ can be a Condorcet winner in some election
 $\elct$ realizing~$\posmtrx$ (II).
 Finally, we demonstrate that if $\posmtrx$ is realizable by
 an election restricted to
 domain $\domain$ in which $\pcand$ is a Condorcet winner,
 then instance $\originstance$ has an exact cover (III).
 
 We begin, however, with a general remark about elections
 realizing $\posmtrx$ restricted to domain $\domain$,
 which we will use throughout the proof.
 Observe that the only vote in $\domain$
 in which the preferred candidate $\pcand$ is in the first position is $\voteforp$
 and, analogously, the only vote in $\domain$ in which 
 the despised candidate $\dcand$ is in the first position is $\voteford$.
 Thus, election realizing $X$ must include
 exactly $4\setcovsize + 1$~votes $\voteforp$ and
 $7\setcovsize - 1$~votes~$\voteford$.
 Hence, we arrive at the following conclusion.
 \begin{claim}\label{claim:real_submatrix}
  In every election $\elct = (C,V)$ realizing matrix $\posmtrx$,
  in which all votes belong to domain $\domain$,
  the voters can be split into two disjoint sets $V_1, V_2 \subseteq V$,
  such that $V_1 \cup V_2 = V$ and:
  \begin{enumerate}
      \item $V_1$ consists of $4\setcovsize +1 1$ votes $\voteforp$
      and $7\setcovsize - 1$ votes \voteford, and
      \item $V_2$ consists of $9\setcovsize$ votes solely from $\domain_2$.
  \end{enumerate}
 \end{claim}
 Building upon \cref{claim:real_submatrix},
 we consider elections $\elct_1 = (V_1,C)$ and $\elct_2 = (V_2,C)$.
 Since in both votes $\voteforp$ and~$\voteford$ positions $3,4,\dots,9\setcovsize +2$
 are occupied by candidates
 $\ecand_1,\dots,\ecand_{3\setcovsize},
 \fcand_1,\cand_1,\dots,\fcand_{3\setcovsize},\cand_{3\setcovsize}$,
 respectively,
 $\elct_1$ realizes matrix $\advmatrix$.
 Hence,
 in the remainder of the proof,
 we mainly focus on the construction of election~$\elct_2$
 restricted to $\domain_2$ that realizes matrix $\constrmatrix$.
 
 (I)
 We construct such election $\elct_2 = (V_2,C)$
 realizing~$\constrmatrix$ as follows. We start with empty~$V_2$ and we add to $V_2$
 one copy of each of the votes
 $\votefromdgroup_1(\uelem_i),\votefromdgroup_2(\uelem_i),\votefromdgroup_3(\uelem_i)$,
 for every $\uelem_i \in \univ$.
 
 Now, let us show that $\elct_2$ indeed realizes matrix $\constrmatrix$.
 For every $\uelem_i \in \univ$,
 the preferred candidate $\pcand$ is at positions
 $6\setcovsize + 2$ in $\votefromdgroup_1(\uelem_i)$,
 $6\setcovsize + 1$ in $\votefromdgroup_2(\uelem_i)$, and
 $6\setcovsize + 3$ in $\votefromdgroup_3(\uelem_i)$.
 Hence, in each of these positions it appears $3\setcovsize$ times in total.
 The same is true for the despised candidate $\dcand$
 (it appears at position
 $6\setcovsize + 1$ in $\votefromdgroup_1(\uelem_i)$,
 $6\setcovsize + 3$ in $\votefromdgroup_2(\uelem_i)$, and
 $6\setcovsize + 2$ in $\votefromdgroup_3(\uelem_i)$,
 for every $\uelem_i \in \univ$).
 Thus, candidates $\pcand$ and $\dcand$ appear at positions
 as indicated by matrix $\constrmatrix$.
 
 For every $i \in [3\setcovsize]$,
 executive candidate $\ecand_i$ appears at position
 $6\setcovsize + 3$ in $\votefromdgroup_1(\uelem_i)$,
 $6\setcovsize + 2$ in $\votefromdgroup_2(\uelem_i)$, and
 $6\setcovsize + 1$ in $\votefromdgroup_3(\uelem_i)$.
 In every other vote, i.e.,
 vote $\votefromdgroup_\ell(\uelem_j)$ for some $\uelem_j \in (\univ \setminus \{
 u_i\})$ and $\ell \in [3]$,
 it appears in either position $6t + 2 + i$,
 if $j < i$,
 or position $6t + 3 + i$,
 if $j > i$.
 Hence, in total, it appears once at positions $6\setcovsize + 1, 6\setcovsize +
 2$, and $6\setcovsize + 3$;
 $3i - 3$ times at position $6t + 2 + i$;
 and $9\setcovsize - 3i$ times at position~$6t + 3 + i$.
 This conforms to matrix $\constrmatrix$.
 
 Finally,
 for every $i \in [3\setcovsize]$, $\uelem_j \in \univ$ and $\ell \in [3]$,
 element and filler candidates, $\fcand_i$ and $\cand_i$, appear at positions $2i-1$ and $2i$, respectively,
 in every vote $\votefromdgroup_\ell(\uelem_j)$,
 unless $j = 1$ and $\ell = 3$,
 where $\fcand_i$ appears at positions $2i$ and $\cand_i$ in $2i-1$.
 Thus, positions of these candidates also agree with matrix $\constrmatrix$.
 This concludes the proof of part (I).

 (II)
 Now, let us show that if~\originstance{}~admits an exact cover~$\excov{}$,
 then there exists election $\elct$ restricted to domain $\domain$
 realizing matrix $\posmtrx$
 such that $p$ is the Condorcet winner.
 Without loss of generality,
 we assume~$\excov{} = \{\sset_1, \sset_2, \ldots, \sset_\setcovsize\}$.
 Building upon~\cref{claim:real_submatrix},
 we focus on constructing $\elct_2 = (C,V_2)$ restricted to domain $\domain_2$
 realizing matrix $\constrmatrix$.
 To this end, we include in $V_2$ the following votes:
 \begin{enumerate}
  \item For each~$\sset \in \excov$, we add voter~$\votefrompgroup(\sset)$, which gives,
  in total, $\setcovsize$ votes;\label{en:cover_votes}
  \item For each~$\uelem_i \in (\univ \setminus \reprcov(\excov))$, we add
  voter~$\votefrompgroup_3(\uelem_i)$, which gives,
  in total, $2 \setcovsize$ votes; \label{en:additional_fill}
  \item For each~$\uelem_i \in \univ$, we add
  voters~$\votefrompgroup_1(\uelem_i)$ and~$\votefrompgroup_2(\uelem_i)$, which
  gives, in total, $6 \setcovsize$~votes. \label{en:all_fill}
 \end{enumerate}
 Let us show that such defined election $\elct_2$
 indeed realizes matrix $\constrmatrix$.
 
 For every $u_i \in \univ \setminus \reprcov(\excov)$,
 the preferred candidate~$\pcand$
 appears at positions 
 $6\setcovsize +1, 6\setcovsize +2, 6\setcovsize +3$
 in votes~$\votefrompgroup_1(u_i)$, $\votefrompgroup_3(u_i)$, and~$\votefrompgroup_2(u_i)$,
 respectively.
 Since each element~$u_i \in \reprcov(\excov)$
 belongs to exactly one $S_j \in K$,
 we also have that $\pcand$
 appears at positions 
 $6\setcovsize +1$, $6\setcovsize +2$, $6\setcovsize +3$ respectively
 in votes
 $\votefrompgroup_1(u_i)$, $\votefrompgroup(S_j)$, and $\votefrompgroup_2(u_i)$.
 Hence, in total, $\pcand$ appears $3\setcovsize$~times
 in each position $6\setcovsize +1$, $6\setcovsize +2$, and $6\setcovsize +3$.
 By analogous reasoning, this is also true for the despised candidate $\dcand$,
 which agrees with position matrix~\constrmatrix.
 
 For each executive candidate $\ecand_i$,
 observe that $\ecand_i$~appears at positions 
 $6\setcovsize +1$, $6\setcovsize +2$, and $6\setcovsize +3$
  in votes
 $\votefrompgroup_1(u_i)$, $\votefrompgroup_3(u_i)$, $\votefrompgroup_2(u_i)$,
 respectively
 (or in votes
 $\votefrompgroup_1(u_i)$, $\votefrompgroup(S_j)$, $\votefrompgroup_2(u_i)$
 if $u_i$ is the representative of some set $S_j \in \excov$).
 In the remaining votes
 it occupies either position~$6t + 2 + i$ or position~$6t + 3 + i$,
 depending on the index of element from $\univ$, as in part (I).
 Thus, again, its positions conform to matrix $\constrmatrix$.
 
 Finally, consider filler candidate~$\fcand_i$ and element candidate~$\cand_i$.
 Observe that $\fcand_i$ is at position $2i-1$ and $\cand_i$ at position~$2i$
 in every vote except $\votefrompgroup(S_j)$ such that $S_j \in \excov$ and $\uelem_i \in S_j$.
 Since there is exactly one such $S_j$ for every $i \in [3\setcovsize]$,
 we obtain that $\fcand_i$ appears $9\setcovsize - 1$ times at position~$2i-1$
 and once at position~$2i$,
 whereas $\cand_i$ appears $9\setcovsize - 1$~times at position~$2i$
 and once at position $2i-1$.
 This agrees with matrix $\constrmatrix$,
 so $\constrmatrix$ is realized by election $\elct_2$.
 
  It remains to show that the preferred candidate $\pcand$ is indeed the Condorcet winner
  in~$\elct$.
  To this end, observe that in votes from $V_1$
  candidate $\pcand$ wins $11\setcovsize$ times
  with each candidate in~$C \setminus \{\pcand,\dcand\}$. However, there is
  only $9\setcovsize$~votes in~$V_2$, so~$\pcand$ cannot lose against any
  candidate from~$C \setminus \{\pcand,\dcand\}$. As per candidate~$\dcand$,
  $\pcand$ wins only $4\setcovsize + 1$ times with~$\dcand$ in votes from~$V_1$.
  Hence, since~$|V_1 \cup V_2| = 20\setcovsize$, we have to show that among votes in~$V_2$,
  candidate~$\pcand$ is preferred over $\dcand$
  at least $6\setcovsize$~times (this means that $\pcand$ wins with $\dcand$
  exactly $10\setcovsize$~times in total).
  Observe that this is exactly the case,
  as $\pcand$ is preferred over $\dcand$ in:
  \begin{itemize}
      \item vote $\votefrompgroup_1(\uelem_i)$, for every $\uelem_i \in \univ$
      ($3\setcovsize$ votes),
      \item vote $\votefrompgroup_3(\uelem_i)$, for every $\uelem_i \in \univ \setminus \reprcov(\excov)$
      ($2\setcovsize$ votes), and
      \item vote $\votefrompgroup(S_j)$, for every $S_j \in \excov$
      ($\setcovsize$ votes).
  \end{itemize}
  This concludes the proof of part (II).
  
 (III) Now, let us assume that
 there exists an election~$\elct = (C,V)$ restricted to domain $\domain$
 realizing~$\posmtrx$ such that the preferred candidate~$\pcand$
 is the Condorcet winner.
 We show that this implies that the original instance~$\originstance$
 admits an exact cover.
 
 From part (II) above we know that
 candidate $\pcand$ wins with every candidate in $C \setminus \{\pcand,\dcand\}$
 at least $11\setcovsize$ times.
 However, only in $4\setcovsize + 1$ votes from $V_1$ it is preferred over $\dcand$.
 The fact that $\pcand$ is the Condorcet winner implies that
 it is preferred over $\dcand$ in at least $10\setcovsize + 1$ votes in $V$.
 Hence, $\pcand$ wins with $\dcand$ in at least $6\setcovsize$ votes in $V_2$.
 
 From the fact that $\elct_2$ realizes matrix $\constrmatrix$,
 we see that in votes from $V_2$
 both $\pcand$ and $\dcand$ occupy only positions
 $6\setcovsize + 1$, $6\setcovsize + 2$, and $6\setcovsize + 3$,
 each $3\setcovsize$ times.
 Thus, in all $3\setcovsize$ votes
 in which $\dcand$ is at position $6\setcovsize + 1$,
 candidate $\dcand$ has to be preferred over $\pcand$.
 However, since in all votes from $V_2$
 candidate $\pcand$ is preferred over $\dcand$ at least $6\setcovsize$~times,
 this implies that $\pcand$ wins with $\dcand$ in all remaining votes
 (in all of them $\dcand$ is not at position $6\setcovsize + 1$).
 This means that in every vote in which
 candidate $\pcand$ is at position $6\setcovsize + 1$ or $6\setcovsize + 2$
 candidate $\dcand$ occupies the next position:
 $6\setcovsize + 2$ or $6\setcovsize + 3$, respectively.
 Therefore, $V_2$ can include only votes denoted with a plus sign:
 $\votefrompgroup_1(u_i)$, $\votefrompgroup_2(u_i)$, $\votefrompgroup_3(u_i)$,
 for each $u_i \in \univ$; and
 $\votefrompgroup(S_i)$, for each $S_i \in \setsfam$.
 
 Now, consider an element candidate $\cand_i$,
 for some fixed $i \in [3\setcovsize]$.
 Since $\elct_2$ realizes matrix~$\constrmatrix$,
 $\cand_i$ appears at position  $2i-1$ exactly once.
 However, among votes denoted with a plus sign,
 $\cand_i$ appears at position  $2i-1$
 only in $\votefrompgroup(S_j)$,
 for these $S_j \in \setsfam$ for which $u_i \in S_j$.
 Let us denote by $\excov$ the collection of such sets $S \in \setsfam$
 for which $\votefrompgroup(S)$ is a vote in election $\elct_2$, i.e.,
 \[
    \excov = \{ S \in \setsfam : \votefrompgroup(S) \in V_2 \}.
 \]
 Then, for every $\uelem_i \in \univ$,
 since in $\elct_2$ candidate $\cand_i$
 appears at position $2i-1$
 exactly once,
 we have that $\uelem_i$
 belongs to exactly one set $S \in \excov$.
 Therefore, $\excov$ is an exact cover.
 This, concludes the proof.
 \endgroup
\end{proof}

\subsection[Proof of Theorem 7]{Proof of \cref{thm:condorcet:condition}}

\condorcetcondition*
\begin{proof}
First, let us prove the condition
and then we will show that it can be checked
in polynomial time.

Let $\elct = (C,V)$ be arbitrary elections with a Condorcet winner,
$c \in C$, and let $X$ be its position matrix
(for notational convenience, by $X_{i,j}$, we will denote number of voters that rank candidate $j \in C$ at position $i$).
Let us also take an arbitrary set of candidates $S \subseteq C \setminus \{c\}$ and number $i \in [m]$.

Now, for each $k \in [i-1]$, by $V_k \subseteq V$ let us denote the set of votes in which candidate $c$ is at position $k$.
Then, by $N_k$ let us denote the number of times candidates from $S$ appear at positions $k+1,k+2,\dots,i$ in votes in $V_k$.
Formally,
\[
    N_k = \sum_{v \in V_k} | \{ j \in S :  k < \pos_v(j) \le i \}|.
\]
Since in each vote, there are total of $i-k$ such positions,
we have
$| \{ j \in S :  k < \pos_v(j) \le i \} | \le i - k$.
Moreover, there can be at most $|S|$ candidates from $S$
at these positions, we get
$| \{ j \in S :  k < \pos_v(j) \le i \} | \le |S|$.
Hence,
\[
    N_k \le
    \sum_{v \in V_k} \min(|S|,i-k) =
    X_{k,c} \cdot \min(|S|,i-k).
\]
Summing for all $k \in [i-1]$ we obtain
\begin{equation}
\label{eq:prop:cw-nnecessary-condition-one:1}
    \sum_{k=1}^{i-1} N_k \le
    \sum_{k=1}^{i-1} X_{k,c} \cdot \min(|S|,i-k).
\end{equation}

Now, by $N_{k,j}$ let us denote the number of votes from $V_k$ in which candidate $j$ appears at positions $k+1,k+2,\dots,i$.
Formally,
\[
    N_{k,j} = | \{ v \in V_k :  k < \pos_v(j) \le i \} |.
\]
Since $c$ is a strong Condorcet winner,
for every candidate $j \in S$
there are at least $\lfloor (n + 1)/2 \rfloor$ votes
in which $c$ is ranked before $j$.
In particular,
if $j$ is ranked at position $i$ or smaller in $X_{j, \le i}$ votes,
then there exist at least
\(
    \max ( 0 , \lfloor (n+1)/2 \rfloor - (n - X_{j, \le i}))
\)
votes in which $j$ is at position $i$ or smaller and $c$ is ranked before $j$.
We count each such vote in one of sets $N_{1,j},\dots,N_{i-1,j}$,
hence we get that
\begin{align*}
    \sum_{k = 1}^i N_{k,j} &\ge
    \max \left( 0 , X_{j, \le i} - \left\lfloor \frac{n-1}{2} \right\rfloor \right) \\ &\ge
    X_{j, \le i} - \left\lfloor \frac{n-1}{2} \right\rfloor\\ &=
    \sum_{k=1}^i X_{k,j} - \left\lfloor \frac{n-1}{2} \right\rfloor.
\end{align*}
Observe that $\sum_{j \in S} N_{k,j}= N_k$.
Therefore, we get that
\begin{align*}
    \sum_{k = 1}^i N_k &\ge
    \sum_{j \in S} \left( \sum_{k=1}^i  X_{k,j}  - \left\lfloor \frac{n-1}{2} \right\rfloor \right) \\ &=
    \sum_{j \in S} \sum_{k=1}^i  X_{k,j} - |S| \cdot \left\lfloor \frac{n-1}{2} \right\rfloor.
\end{align*}
Combining this with inequality~\eqref{eq:prop:cw-nnecessary-condition-one:1} we obtain the thesis.

Now, let us focus on proving that the condition can be checked in polynomial time.
To this end, let us fix $i \in [m]$
and sort the candidates in $C \setminus \{c\}$
in the order of how often they appear in first $i$ positions.
Formally, let $o$ be a strict linear order on $C \setminus \{c\}$
such that $j \succ_o j'$
if and only if either
\[
    \sum_{k=1}^i X_{k,j} > \sum_{k=1}^i X_{k,j'}
\]
or
\[
    \sum_{k=1}^i X_{k,j} = \sum_{k=1}^i X_{k,j'}
    \quad \mbox{and} \quad
    j > j'.
\]
Now, for each $j \in C \setminus \{c\}$,
let us denote by $S_{i,j}$ a set of candidates other than $c$
that is $j$ and candidates before $j$ in the order $o$.
Formally,
\[
    S_{i,j} = \{ k \in C \setminus \{c\} : k \succeq_o j\}.
\]
Observe that for every $i \in [m]$ and $S \subseteq C \setminus \{c\}$
it holds that the left hand side of condition for set $S$
is smaller or equal to the left hand side of condition for set of
$|S|$ candidates other than $c$ most frequently appearing in first $i$ positions,
i.e.,
\[
    \sum_{j \in S} \sum_{k=1}^i  X_{k,j} \le
    \sum_{j \in S_{i,\pos_o(|S|)}} \sum_{k=1}^i  X_{k,j}.
\]
Hence, it suffices to check the condition only for sets $S_{i,j}$
for all $i \in [m]$ and all $j \in C \setminus \{c\}$,
which can be done in polynomial time.
\end{proof}

\subsection{Experiments}

\subsubsection[No. Possible Condorcet Winners]{Computing the Number of Possible Condorcet Winners}
As described in the main body, given a position matrix, we computed the number of different candidates which are the Condorcet winner in some realization of the matrix. 
To solve this problem, we derive an Integer Linear Program (ILP) for the following closely related problem: 
Given an $m\times m$ position matrix $X$ over a candidate set $C$ (where rows and columns sum up to $n$) and a candidate $c^*\in C$, is there an election realizing $X$ in which $c^*$ is a Condorcet winner.

To model this problem as an ILP, we introduce for each $c\in C$, $i\in [m]$, and $k\in [n]$ a binary variable $x_{c,k,i}$. 
Setting $x_{c,k,i}$ to true corresponds to putting candidate $c$ on position $i$ in vote $k$. 
Let $y$ be an $n$-dimensional vector which for each $i\in [m]$ contains $X_{i,c^*}$-times the number $i$. 
Because the ordering of votes is clearly irrelevant, we can start by fixing the position in which $c^*$ is ranked in each vote:
\begin{align*}
    x_{c^*,k,y_k}=1.
\end{align*}

To enforce that every position in each vote is taken by exactly one candidate and that each candidate appears in each vote exactly once we add the following constraints: 
\begin{align*}
    &\sum_{c\in C} x_{c,k,i} =1, \qquad &\forall i\in [m], k\in [n]\\
    &\sum_{i\in [m]} x_{c,k,i} =1 & k\in [n], c\in C
\end{align*}
Moreover, we add constraints enforcing that each candidate appears in every position as often as specified in its position vector: 
\begin{align*}
    \sum_{k\in [n]} x_{c,k,i}=X_{i,c}, \qquad i\in [m], c\in C
\end{align*}
Lastly, we ensure that candidate $c^*$ wins the pairwise comparison against each other candidate. 
We can easily enforce this because we know in each vote the position on which $c^*$ appears: 
\begin{align*}
    \sum_{k\in [n]: i\in [y_k-1]} x_{c,k,i}\leq \left\lceil\frac{n+1}{2}\right\rceil, \qquad c\in C
\end{align*}

\end{document}